\documentclass[11pt,a4paper]{article}

\usepackage{authblk}
\usepackage{fullpage}
\usepackage{times}
\usepackage{dsfont}
\usepackage{soul}
\usepackage{url}
\usepackage[hidelinks,pagebackref]{hyperref}
\usepackage[utf8]{inputenc}
\usepackage[small]{caption}
\usepackage{graphicx}
\usepackage{amsmath,amssymb}
\usepackage{amsthm}
\usepackage{booktabs,siunitx}
\usepackage{algorithm}
\usepackage[noend]{algpseudocode}
\usepackage{xcolor,makecell}
\usepackage{tabularx,environ}
\usepackage{xspace}
\usepackage{amsfonts}
\usepackage{subcaption}
\usepackage[switch,modulo]{lineno}
\usepackage[sort,numbers]{natbib}
\usepackage{tikz}
\usetikzlibrary{arrows}
\usetikzlibrary{calc, fit}
\usepackage{blkarray, bigstrut}
\newcommand{\fcsr}{\ensuremath{f^{\text{CSR}}}}
\newcommand{\flsr}{\ensuremath{f^{\text{LSR}}}}
\newcommand{\fst}{\ensuremath{f^{(s,t)}}}

\newtheorem{observation}{Observation}
\newtheorem{fact}{Fact}

\newcommand{\const}{\textsc{Const}-}
\newcommand{\dest}{\textsc{Dest}-}
\newcommand{\mixed}{\textsc{Const+Dest}-}
\newcommand{\exact}{\textsc{Exact-}\xspace}

\newcommand{\AgentBribery}[1]{#1~\textsc{Agent Bribery}\xspace}
\newcommand{\LinkBribery}[1]{#1~\textsc{Link Bribery}\xspace}
\newcommand{\pAgentBribery}[1]{#1~\textsc{\$Agent Bribery}\xspace}
\newcommand{\pLinkBribery}[1]{#1~\textsc{\$Link Bribery}\xspace}
\newcommand{\pnAgentBribery}[1]{#1~\textsc{(\$)Agent Bribery}\xspace}

\DeclareMathOperator*{\argmin}{arg\,min}
\DeclareMathOperator*{\opin}{in}
\DeclareMathOperator*{\opout}{out}
\DeclareMathOperator*{\cheap}{{\ensuremath{\zeta}}}

\algnewcommand{\IIf}[1]{\State\algorithmicif\ #1\ \algorithmicthen}
\algnewcommand{\EndIIf}{\unskip}

\newtheorem{theorem}{Theorem}
\newtheorem{proposition}{Proposition}

\newtheorem{corollary}{Corollary}

\newtheorem{lemma}{Lemma}

\tikzstyle{vertex}=[draw, circle, fill, inner sep = 2pt]
\tikzstyle{squared-vertex}=[draw, fill, inner sep = 2pt]

\usetikzlibrary{arrows.meta, shapes.geometric}

\usepackage{framed}
\usepackage{tabularx}

\newlength{\RoundedBoxWidth}
\newsavebox{\GrayRoundedBox}
\newenvironment{GrayBox}[1]%
   {\setlength{\RoundedBoxWidth}{.93\columnwidth}
    \def\boxheading{#1}
    \begin{lrbox}{\GrayRoundedBox}
       \begin{minipage}{\RoundedBoxWidth}}%
   {   \end{minipage}
    \end{lrbox}
    \begin{center}
    \begin{tikzpicture}%
       \node(Text)[draw=black!20,fill=white,rounded corners,inner sep=2ex,text width=\RoundedBoxWidth]
             {\usebox{\GrayRoundedBox}};
        \coordinate(x) at (current bounding box.north west);
        \node [draw=white,rectangle,inner sep=3pt,anchor=north west,fill=white]
        at ($(x)+(6pt,.75em)$) {\boxheading};
    \end{tikzpicture}
    \end{center}}

\newenvironment{defproblemx}[1]{\noindent\ignorespaces%
                                \FrameSep=6pt%
                                \parindent=0pt%
                \begin{GrayBox}{#1}%
                \begin{tabular*}{\columnwidth}{!{\extracolsep{\fill}}@{\hspace{.1em}} >{\itshape} p{1.cm} p{0.86\columnwidth} @{}}%
            }{
                \end{tabular*}%
                \end{GrayBox}%
                \ignorespacesafterend
            }

\newcommand{\defProblemTask}[3]{%
  \begin{defproblemx}{#1}
    Input: & #2 \\
    Task: & #3
  \end{defproblemx}
}

\title{Fine-Grained View on Bribery for Group Identification\thanks{A 
preliminary version of this article appeared in the \emph{Proceedings of the 
Twenty-Ninth International Joint Conference on
               Artificial Intelligence (IJCAI~2020)}, ijcai.org, 
			pages~67--73, 
		2020.}}

\author[1]{Niclas Boehmer} 
\author[2]{Robert Bredereck}
\author[3]{Du{\v
		s}an Knop}
\author[4]{Junjie
	Luo}

\affil[1]{\small
	Algorithmics and
	Computational Complexity, TU
	Berlin, Germany, niclas.boehmer@tu-berlin.de}
\affil[2]{\small
	Humboldt-Universität zu Berlin, Germany,
	robert.bredereck@hu-berlin.de}
\affil[3]{\small
	Czech Technical University in Prague, Czech Republic,
	dusan.knop@fit.cvut.cz}
\affil[4]{\small
Nanyang Technological University, Singapore, junjie.luo@ntu.edu.sg}

\date{\today}

\usepackage{etoolbox}
\usepackage{marginnote}
\usepackage[textsize=scriptsize]{todonotes} 
\newcommand{\mytodo}[2]{\xspace}
\newcommand{\myrevtodo}[2]{{%
		\let\marginpar\marginnote
		\reversemarginpar
		\renewcommand{\baselinestretch}{0.8}%
		\xspace}}
\newcommand{\myinlinetodo}[2]{\xspace}
\newcommand{\registerAuthor}[3]{%
	\expandafter\newcommand\csname #2com\endcsname[1]{\mytodo{#3}{\textsc{#2}:
	##1}}%
	\expandafter\newcommand\csname
	#2revcom\endcsname[1]{\myrevtodo{#3}{\textsc{#2}: ##1}}%
	\expandafter\newcommand\csname
	#2inline\endcsname[1]{\myinlinetodo{#3}{\textsc{#2}: ##1}}%
	\expandafter\newcommand\csname
	#2inlineLater\endcsname[1]{\lv{\myinlinetodo{#3}{\textsc{#2}: ##1}}}%
}
\registerAuthor{Niclas}{nb}{orange!50!white}
\registerAuthor{Robert}{rb}{yellow}
\registerAuthor{Dusan}{dk}{green!40!blue!25!white}
\registerAuthor{Junjie}{jl}{gray!10!white}

\begin{document}

\maketitle

\begin{abstract}
  Given a set of agents qualifying or disqualifying each other,
  group identification is the task of identifying a \emph{socially qualified}
  subgroup of agents.
  Social qualification depends on the specific rule used to aggregate individual qualifications.
  The classical bribery problem in this context asks how many agents need to
  change their qualifications in order to change the outcome in a certain way.

  Complementing previous results showing poly\-no\-mial-time solvability or NP-hardness of bribery for various social rules
  in the constructive (aiming at making specific agents socially qualified) or
  destructive (aiming at making specific agents socially disqualified)
  setting, we provide a comprehensive picture of the parameterized computational complexity landscape.
  Conceptually, we also consider a more fine-grained concept of bribery cost,
  where we ask how many single qualifications need to be changed, nonunit
  prices for different bribery actions, and a more
  general bribery goal that combines the constructive and destructive setting.
\end{abstract}

\section{Introduction}

The University of Actual Truth (UAT) was paralyzed for months due to a heavy
dispute of the scientists about who belongs to the group of \emph{true scientists}.
After some literature research on group identification, they asked every
scientist to report who they believe is qualified for being a true scientist.
Based on these individual qualifications, they applied several group identification rules to find
the group of true scientists.
Unfortunately, each rule either decides that nobody is a true scientist
or that all are true scientists.
It is, however, obvious to everyone that the group of true scientists must
be a proper nonempty subset of scientists.

As the next step, the UAT scientists computed their ``degree of truthfulness'' as follows. First of all, every department was invited to submit a proposal specifying  who they believe the true scientists are.
To evaluate and compare every group being proposed by someone, they computed two different quality measures.
For the rules that identified nobody as a true scientist,
they defined the ``truth distance'' as the minimum number of scientists whose qualifications
would need to be changed to make the proposed group part of the true scientists.
For the rules that initially identified everyone as a true scientist,
they defined the ``margin of truth'' as the minimum number
of scientists whose qualifications would need to be changed to make no one from the
proposed group a true scientist.

The above example describes a problem that appears in many situations
where one needs to identify a socially qualified group of agents based only on
the agents' pairwise qualifications.
To solve this task, group identification rules have been
developed~\citep{DBLP:journals/LeA/KasherR1997,DBLP:journals/jet/SametS03}:
Among the most important rules studied in the literature, which are studied in
this paper, are the \emph{consent
	rule} \citep{DBLP:journals/jet/SametS03} and the two iterative rules:
\emph{consensus-start-respecting rule} (CSR,
\citep{DBLP:reference/JCI/Kasher1993}) and \emph{liberal-start-respecting rule}
(LSR, \citep{DBLP:journals/LeA/KasherR1997}). Using the consent rule,
which is parameterized by integers $s$ and $t$, each agent qualifying itself
is socially qualified if and only if at least $s$ agents qualify it,
and each agent disqualifying itself is socially disqualified if and only if at
least $t$ agents disqualify it. In the iterative rules, some criterion is
used to determine an initial set of socially qualified agents, which is then
iteratively extended by adding all agents who are qualified by at least
one agent from the set of already socially qualified agents until convergence
is reached. Under CSR, an agent is initially socially qualified if it is
qualified by all agents, while, under
LSR, all agents qualifying themselves are initially socially qualified.

Despite the simplicity of our example, it illustrates an important aspect of
group identification rules:
Group identification rules provide only a binary decision about membership to
a specific group, while multiple degrees of certainty about membership may be desired.
In extreme cases (as in the example), the identified group might contain indeed
too many or too few agents.
The distances that ``solve'' this issue in our example are concepts known
in the literature, but usually motivated from a different viewpoint:
``truth distance'' corresponds to constructive bribery and
``margin of truth'' corresponds to destructive bribery.
Herein, the classical bribery model assumes an external agent (with full knowledge
over the individual qualifications) that aims to influence
the outcome of the group identification process by convincing a limited number of
agents to change their qualifications in a certain way to achieve some
goal.
While the assumptions behind the classical bribery motivation may be
questionable in the context of group identification,
we emphasize that computing bribery costs as a quality measure is very natural
and useful in practice,
as illustrated in our example.

\subsection{Our Contributions}
In this paper, we provide a more fine-grained view on computing bribery costs
for group identification rules in three ways.
First, we allow to combine \emph{constructive} and \emph{destructive bribery}.
In particular, we allow to specify two sets~$A^+$ and~$A^-$ of agents
that must be (resp., must not be) socially qualified after the bribery action.
This includes as a special case \emph{exact bribery}, where one can specify the
final socially qualified subgroup of agents.
Second, we consider a more fine-grained concept of bribery costs called
\emph{link bribery}, where one counts the number of individual qualifications
that need
to be changed. So far, in the classical model, which we call \textit{agent
	bribery}, the number of agents that alter their qualifications is counted.
Third, we consider priced versions of our bribery problems, i.e., in link
bribery, each qualification has a price of being changed and, in agent bribery,
each agent has a price of being bribed. Notably, the priced and unpriced
versions of all considered problems have the same complexity, i.e., all our
hardness results hold for the unpriced problems, while all algorithmic results
apply to the priced problems.
Fourth, we complement the classical (P vs.\ NP) computational complexity
landscape by providing a comprehensive analysis of the parameterized complexity,
focusing on naturally and well-motivated parameters such as the bribery cost
and the sizes of the sets~$A^+$ and~$A^-$ as well as rule-specific parameters.
We refer to \autoref{t:Sum} and \autoref{fig:fst-c-agent-overview} for an
overview of our results and to
\autoref{se:prelim} for formal
definitions of the rules and parameters. Note that the results for constructive
bribery depicted in \autoref{fig:fst-c-agent-overview} analogously hold for
destructive bribery
with switched roles of $s$ and $t$ (see
\citep{DBLP:journals/aamas/ErdelyiRY20} and \autoref{le:1}).

\begin{table*}[t]

	\centering
	\begin{tabular}{@{}cllllllllllllll@{}}
		\toprule
		&
		\multicolumn{2}{c}{\fcsr/\flsr} &
		\phantom{a}&
		\multicolumn{2}{c}{\fst}
		\\
		\cmidrule{2-3}
		\cmidrule{5-6}
		& Agent & Link && Agent & Link\\
		\midrule
		Const & P ($\dagger$) & NP-c. (Th. \ref{thm:fsr-c-arc}) && NP-c.
		($\dagger$) & P (O. \ref{ob:link-mixd-fst})\\
		&  & FPT wrt. $|A^+|$ (Th. \ref{thm:fsr-c-fpt-arc}) &&  & \\
		&  & W[1]-h. wrt. $\ell$ (Th. \ref{thm:fsr-c-arc}) &&  & \\
		\midrule
		Dest & P ($\dagger$) & P (Th. \ref{thm:flsr-d-arc})&& NP-c.
		$(\dagger)$ & P (O. \ref{ob:link-mixd-fst}) \\
		\midrule
		Const+Dest & P (Th. \ref{thm:fsr-cd-agent}) & NP-c. (Th.
		\ref{thm:fsr-c-arc}) && NP-c. (Ob. \ref{ob:agent-mixed-fst}) & P
		(O. \ref{ob:link-mixd-fst}) \\
		\midrule
		Exact & P (Co. \ref{co:flsr-e-arc}) & P (Th. \ref{thm:flsr-d-arc})
		&& NP-c. (Ob. \ref{ob:agent-mixed-fst}) & P (O.
		\ref{ob:link-mixd-fst}) \\
		\bottomrule
	\end{tabular}
	\caption{Overview of all our complexity results except for our
		parameterized
		complexity results for consent rules (see
		\autoref{fig:fst-c-agent-overview}). Results with a
		$\dagger$
		were proven by
		\protect\citet{DBLP:journals/aamas/ErdelyiRY20}. All
		polynomial-time and fixed-parameter tractability results also
		hold for the priced versions of the considered problems.}
	\label{tab:Sum1}
	\label{t:Sum}
\end{table*}

\subsection{Related Work}
\citet{DBLP:journals/jair/FaliszewskiHH09} introduced bribery problems to the
study of elections by studying the problem of making a given candidate a winner
by changing the preferences of at most a given number of voters (a problem
being closely related to constructive agent bribery in our setting). Since
then,
multiple variants of bribery differing in the
goal and the pricing of a bribery have been proposed (see a survey by
\citet{DBLP:reference/choice/FaliszewskiR16}) and bribery problems have
also been studied in the context of other collective decision problems
\citep{DBLP:conf/aldt/BaumeisterER11,DBLP:conf/sagt/BoehmerBHN20}. For example,
\citet{DBLP:journals/jair/FaliszewskiHHR09}       introduced
\textit{microbribery}, where the manipulator pays per flip in the preference
profile of the given election. Microbribery is conceptually closely related to
link bribery in the context of our problem. Furthermore, while
\citet{DBLP:conf/aldt/BaumeisterER11} already considered a variant of exact
bribery in the context of judgment aggregation, we are not aware of any
applications of the combined setting of constructive and destructive bribery
that we propose in this paper.

Despite different initial motivations, bribery in elections is closely related
to the concept \textit{margin of victory}, where the goal is to measure the
robustness of the outcome of an election or the ``distance'' of a candidate
from winning the election
\citep{DBLP:conf/uss/Cary11,DBLP:conf/uss/MagrinoRS11}.
While both concepts have been mostly studied separately, some authors have
developed a unified framework
\citep{DBLP:conf/sigecom/Xia12,DBLP:conf/atal/FaliszewskiST17,DBLP:journals/corr/abs-2010-09678}.

Initially, the group identification problem has been mainly studied from a
social choice perspective by an axiomatic analysis of the problem and some
social rules (see, e.g., the works of
\citet{DBLP:series/sfsc/Dimitrov11},
\citet{DBLP:journals/LeA/KasherR1997}, and \citet{DBLP:journals/jet/SametS03}).
Possible applications of the group identification problem range
from the identification of a collective identity
\citep{DBLP:journals/LeA/KasherR1997} to the endowment of rights
with social implications \citep{DBLP:journals/jet/SametS03}.

Recently, \citet{DBLP:journals/aamas/YangD18} and
\citet{DBLP:journals/aamas/ErdelyiRY20} initiated the study of manipulation by
an external agent both in the context of bribery and control in a group
identification problem for the three mentioned
social rules.  \citet{DBLP:journals/aamas/YangD18} considered the complexity of
agent deletion, insertion, and partition for constructive control, while
\citet{DBLP:journals/aamas/ErdelyiRY20} extended their studies to destructive
control. Moreover, \citet{DBLP:journals/aamas/ErdelyiRY20} analyzed the
complexity of constructive agent bribery and destructive agent bribery. They
proved that for both destructive and constructive bribery, for
CSR and LSR, the related computational problems are polynomial-time solvable.
Moreover, for constructive bribery, they proved that the computational problems
for consent rules with~$t=1$ are also polynomial-time solvable. On the other
hand, for $t\geq2$ and $s\geq 1$, constructive bribery is already NP-complete.
Finally, \citet{DBLP:journals/aamas/ErdelyiRY20} established a close
relationship between constructive and destructive bribery for consent rules by
proving that every
constructive bribery problem  can be converted
into a destructive bribery problem  by switching $s$ and $t$ and flipping
all qualifications.

Independently to this paper (and in parallel to our conference version),
\citet{DBLP:conf/atal/Erdelyi020} studied partly overlapping questions.
In particular, they investigated the computational complexity of constructive
and exact bribery for all three social rules considered in this paper assuming
that the briber pays per modified qualification, which is equivalent to our link
bribery cost model.
In essence, the results of \citet{DBLP:conf/atal/Erdelyi020} overlap with ours in
the unpriced version of the second part of \autoref{thm:flsr-d-arc}, the
NP-hardness from \autoref{thm:fsr-c-arc} and the unpriced version of
\autoref{ob:link-mixd-fst}.
In addition, \citet{DBLP:conf/atal/Erdelyi020} considered how the computational
complexity of these problems changes if one requires that every agent qualifies
exactly $r$ agents before and after the bribery.

The group identification problem is formally related to the multiwinner voting
\citep{DBLP:reference/trends/FaliszewskiS17}. However,
multiwinner voting is of a different flavor both in terms of intended
applications and studied rules. Nevertheless, formally, group
identification is equivalent to approval-based multiwinner voting with a
variable number of winners where the set of voters and candidates coincide.
While (approval-based) multiwinner voting with a variable number of winners has
been studied by, for example, \citet{DBLP:journals/scw/DuddyPZ16},
\citet{DBLP:journals/td/Kilgour2016},
\citet{DBLP:journals/corr/abs-2005-07094}
and
\citet{DBLP:journals/corr/abs-1711-06641}, this specific setting has never been
studied
from a voting perspective. Moreover, so far, the work on bribery for
(approval-based) multiwinner elections is limited to the setting where the
number of winners is fixed
\citep{DBLP:conf/aaai/BredereckFNT16,DBLP:conf/atal/FaliszewskiST17}. Notably,
similar
to our motivation, \citet{DBLP:conf/atal/FaliszewskiST17} also studied how to
measure the margin of victory in approval-based multiwinner elections through the
lens of bribery.

\subsection{Organization}
In \autoref{se:prelim}, we formally define the different computational
questions we examine and provide background on parameterized complexity
theory and some graph algorithms that we later use in our algorithms.
Subsequently, in \autoref{se:itrules}, we present our results for iterative
rules dealing with agent bribery first (Subsection \ref{sub:itab}) and
afterwards
turning to link bribery (Subsection \ref{sub:itlb}). In
\autoref{se:cons},
we consider the consent rule. We start by observing that there exists a simple
algorithm for (priced) link bribery for all our goals and subsequently conduct
a
detailed analysis of constructive link bribery in Subsection \ref{sub:itclb}
before
we explain how the results of this analysis extend to the other goals in
Subsection
\ref{sub:itdlb}. We conclude in \autoref{se:conc} with a summary
of our results and
multiple pointers for possibilities for future work.

\usetikzlibrary{shapes}

\begin{figure*}
\begin{center}
\begin{tikzpicture}
\tikzstyle{para} = [inner ysep=4pt, rectangle, rounded corners, minimum width=2cm, minimum height=0.8cm, text centered, draw=black, thick]

\tikzstyle{para2} = [inner ysep=4pt, rectangle, rectangle split, rectangle split parts=2, rounded corners, minimum width=3cm, minimum height=0.8cm, text centered, draw=black]

\tikzstyle{split-box} = [
rectangle, rounded corners, rectangle split, rectangle split horizontal, rectangle split parts=2, rectangle split part fill={yellow!50,cyan!50}, minimum width=2 cm, minimum height=0.8 cm, text centered, draw=black, thick]

\def\th{0.7cm};   
\def\bl{0.5};

\node (stl) {$s+t+\ell$};

\node (st) [below = 1.5 of stl] {$s+t$};
\node (sl) [below left = 1.2 and 3.5  of stl] {$s+\ell$};
\node (tl) [below right = 1.2 and 3.5 of stl] {$t+\ell$};

\node (s) [below = 2.5 of sl] {$s$};
\node (l) [below = 1.5  of st] {$\ell$};
\node (t) [below = 2.5 of tl] {$t$};

\node (stl-box) [split-box,below=1pt of stl] {\nodepart[text width =1.5cm]{one} XP \nodepart[text width =3.04cm]{two} W[1]-hard (Th. \ref{thm:fst-c-agent-t=1-s})};

\node (st-box) [para, below = 1pt of st, fill=red!50] {Para-NP-hard ($\dagger$)};

\node (sl-box) [para, below = \bl of sl, text width = 4.8cm, text height = \th, fill=cyan!50, thick] {
W[1]-hard even if $t=1$ (Th. \ref{thm:fst-c-agent-t=1-s})};
\node [split-box,below=1pt of sl,thick] {\nodepart[text width =1.55cm]{one} XP \nodepart[text width =3.04cm]{two} W[2]-hard (Th. \ref{thm:fst-c-agent-l})};

\node (tl-box) [para, below = \bl of tl, text width = 4.8cm, text height = \th, fill=green!50, thick] {FPT if $s$ is a constant (Th. \ref{th:fst-lt-fpt})};
\node [split-box,below=1pt of tl,thick] {\nodepart[text width =1.55cm]{one} XP \nodepart[text width =3.04cm]{two} W[1]-hard (Th. \ref{thm:fst-c-agent-t=1-s})};

\node [para, below = 1pt of s, fill=red!50] {Para-NP-hard ($\dagger$)};
\node [para, below = 1pt of t, fill=red!50] {Para-NP-hard ($\dagger$)};

\node [para2, below = \bl of l, text width = 5cm, text height = \th, rectangle split part fill={cyan!50,green!50}, thick] {
 W[1]-hard even if $t=1$ (Th. \ref{thm:fst-c-agent-t=1-s}) \\ W[2]-hard even if $s=1$ (Th. \ref{thm:fst-c-agent-l})
 \nodepart{second}
 FPT if $s,t$ are constants (Th. \ref{th:fst-lt-fpt})
};
\node [split-box,below=1pt of l,thick] {\nodepart[text width =1.55cm]{one} XP \nodepart[text width =3.24cm]{two} W[2]-hard (Th. \ref{thm:fst-c-agent-l})};

\draw (s) -> (sl-box);
\draw (l) -> (sl-box);
\draw (s) -> (st-box);
\draw (t) -> (st-box);
\draw (l) -> (tl-box);
\draw (t) -> (tl-box);
\draw (sl) -> (stl-box);
\draw (st) -> (stl-box);
\draw (tl) -> (stl-box);

\end{tikzpicture}
\end{center}
\caption{Parameterized analysis of \const \AgentBribery{\fst}.
Upper and lower bounds on the complexity of this problem with respect to
each
parameter are shown in the first line, followed by some special cases. The FPT
result from \autoref{th:fst-lt-fpt} also holds for \const \pAgentBribery{\fst}.
The ``XP'' results for all parameter combinations containing $\ell$ are trivial
and also hold for the priced version.
Results with a $\dagger$
     were proven by \protect\citet{DBLP:journals/aamas/ErdelyiRY20}.
Additionally, in \autoref{thm:fst-c-agent-sizeAplus}, we prove that \const
\pAgentBribery{\fst} is FPT wrt. $|A^+|$.}
\label{fig:fst-c-agent-overview}
\end{figure*}
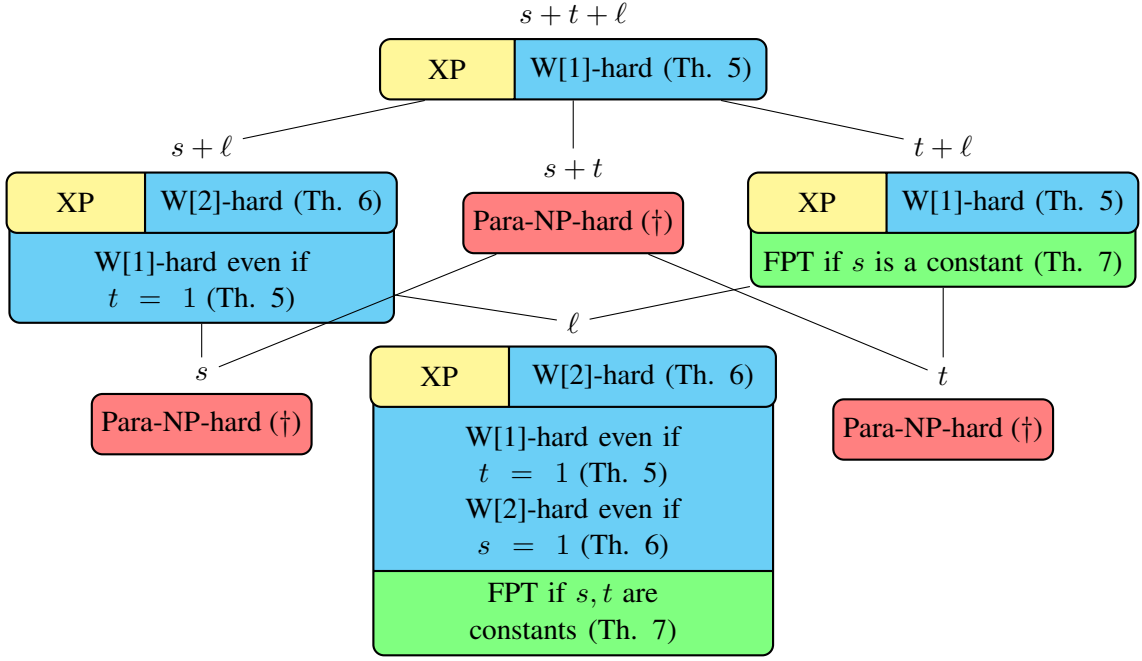

\section{Preliminaries}
\label{se:prelim}
In this section, we start by formally introducing the group identification
problem and all considered social rules (Subsection \ref{sub:GI}).
Subsequently, we
formally define all bribery goals and bribery cost models we analyze in this
paper (Subsection \ref{sub:BV}).
Finally, in Subsection \ref{sub:GTPC}, we introduce some graph theoretical
notations, describe some common
graph problems and the respective algorithms we use in our paper, and provide
a brief introduction to parameterized complexity analysis.

\subsection{Group Identification}\label{sub:GI}
Given a set of agents~$A=\{a_1,\dots,a_n\}$ and a so-called \textit{qualification profile}~\mbox{$\varphi \colon A \times A \to \{-1,1\}$}, the group identification problem asks to return a subset of \textit{socially qualified} agents using some \textit{social rule}~$f$. We write $f(A,\varphi)$ to denote the set of agents that are socially qualified in the group identification problem $(A,\varphi)$ according to $f$.
All agents which are not socially qualified are called \textit{socially
disqualified}. For two agents~$a,a'\in A$, we say that~$a$
\textit{qualifies}~$a'$ if~$\varphi(a,a')=1$; otherwise, we say that~$a$
\textit{disqualifies}~$a'$. For each agent~$a\in A$,
let~$Q_\varphi^+(a)=\{a'\in A\mid \varphi(a',a)=1\}$ denote the set of agents
qualifying~$a$ and~$Q_\varphi^-(a)=\{a'\in A\mid \varphi(a',a)=-1\}$ the set of
agents disqualifying~$a$ in $\varphi$. We omit the subscript $\varphi$ if it
is clear from context. Let $A^*=\{a\in A\mid \forall a'\in A:
\varphi(a',a)=1\}$ be the set of agents who are qualified by everyone including
themselves. For
every group identification problem,~$A$ and~$\varphi$ induce a so-called
(directed)
\textit{qualification graph}~$G_{A,\varphi}=(A,E)$ with~$(a,a')\in E$ if and
only if~$\varphi(a,a')=1$. For two agents~$a$ and $a'$, we say that there
exists a
path from~$a$ to~$a'$ in~$(A,\varphi)$ if there exists a path from~$a$ to~$a'$
in~$G_{A,\varphi}$.

Now, we define social rules considered in this paper:
For the \emph{liberal-start-respecting rule} (\flsr), we start with the
set~$K_1 = \left\{ a \in A \mid \varphi(a,a)=1  \right\}$ and compute the set
of socially qualified agents iteratively for~$i = 2, \ldots$ using
\begin{equation} \label{eq::rule}
 K_i = \left\{ a \in A \mid  \exists a' \in K_{i - 1} : \varphi(a',a)
 =1 \right\}
 \,.
\end{equation}
Notice that we always have~$K_{i - 1} \subseteq K_{i}$.
We stop the process when~$K_{i - 1} = K_{i}$ and output~$K_{i}$.

For the \emph{consensus-start-respecting rule} (\fcsr), we start with the set
$K_1=A^*$
and for $i = 2, \ldots$ we use Equation~\eqref{eq::rule} to compute iteratively
the set of socially qualified agents. Note that it is also possible to compute
the set of socially qualified agents under~\flsr~and~\fcsr~as the set of agents
that correspond to vertices in the qualification graph that are reachable from
vertices with a self-loop and vertices with incoming arcs from all vertices,
respectively.

The \emph{consent rule} (\fst) with parameters~$s$ and~$t$ with~$s+t\leq n+2$
determines the set of socially qualified agents as follows: If~$\varphi(a,a)=1$
for an agent~$a \in A$, then~$a$ is socially qualified if and only if
$|Q^+(a)| \ge s \,$.
If~$\varphi(a,a)= -1$ for an agent~$a \in A$, then $a$ is socially disqualified
if and only if~$|Q^-(a)| \ge t \,.$ Note that the constraint~$s+t\leq n+2$ is
chosen in a
way that an agent who qualifies itself but is socially
disqualified cannot become socially qualified by simply disqualifying
itself.

\subsection{Bribery Variants and Costs} \label{sub:BV}
In the most general form of bribery, which we call
\textsc{Constructive+De\-struc\-tive} (\textsc{Const.+Dest.}) bribery, we are
given a set $A$ of agents, a qualification profile $\varphi$, and a social 
rule~$f$
together with two groups of agents~$A^+$ and~$A^-$ and a budget~$\ell$.
The task is then to alter the qualification profile $\varphi$ such that in the
altered profile $\varphi'$ all agents in~$A^+$ are socially qualified,
i.e.,~$A^+\subseteq f(A,\varphi')$, and all agents in~$A^-$ are socially
disqualified, i.e., $A^-\subseteq A\setminus f(A,\varphi')$. The cost of the
bribery (computed as specified below) is not allowed to exceed~$\ell$.

We also consider the following three special cases of \mixed \textsc{Bribery}:
\begin{description}
	\item[\textsc{Constructive} (\textsc{Const}.)]
	Given a set~$A^+ \subseteq A$ of agents, find a bribery such that $A^+\subseteq f(A,\varphi')$.
	\item[\textsc{Destructive} (\textsc{Dest}.)]
	Given a set~$A^- \subseteq A$ of agents, find a bribery such that $A^-\subseteq A\setminus f(A,\varphi')$.
	\item[\textsc{Exact}]
	Given a set~$A^+ \subseteq A$ of agents, find a bribery such that $A^+=f(A,\varphi')$.
\end{description}

We now specify the cost of a bribery.
As the most important and natural special case, we consider \emph{unit prices}.
In \textsc{Agent Bribery} (with unit prices) the cost of a bribery
is equal to the number of agents whose opinions are modified. Consequently, we
ask whether it is possible to achieve the specified goal by altering the
preferences of at most~$\ell$ agents, where the briber is allowed to change the
preferences of each agent in an arbitrary way.
In \textsc{Link Bribery} (with unit prices), the cost of
a bribery is equal to the number of single qualifications changed.
Therefore, we ask whether it is possible to achieve the specified goal
by altering at most~$\ell$ qualifications, that is, flipping at most~$\ell$
entries in~$\varphi$.

More generally, we consider \emph{priced variants} (denoted by a ``\$'' as prefix in the name)
of our problems where we are additionally given a price function~$\rho$.
For \textsc{\$Agent Bribery}, $\rho$ assigns to each agent~$a$
a positive integer price~$\rho(a)$. For a subset of agents $A'\subseteq A$, we
write $\rho(A')$ to denote $\sum_{a\in A'} \rho(a)$.
For \textsc{\$Link Bribery}, the price function $\rho$ assigns to each ordered
agent pair~$(a,a')\in A\times A$
a positive integer price $\rho((a,a'))$.
The costs of a bribery for \textsc{\$Agent Bribery} is the sum of prices
assigned to the bribed agents~$A' \subseteq A$, i.e., $\sum_{a \in A'}
\rho(a)$.
The costs of a bribery for~\textsc{\$Link Bribery} is the sum of prices
assigned to the modified qualifications~$M' \subseteq A \times A$,
i.e., $\sum_{(a,a') \in M'} \rho((a,a'))$.
We always assume that the maximum cost of an action, that is, the maximum value
of~$\rho$
is bounded by a polynomial in the number of agents.
Note that \textsc{Agent Bribery} (resp., \textsc{Link Bribery}) is equivalent to
\textsc{\$Agent Bribery} (resp.~\textsc{\$Link Bribery})
with~$\rho$ assigning price~$1$ to every agent (resp., to every agent pair).

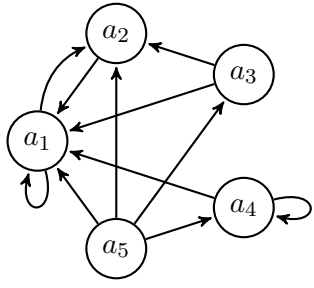
\begin{figure}
	\centering
	\begin{tikzpicture}[->,>=stealth',scale=0.6,shorten
	>=1pt,auto,node
	distance=4cm,
	thick,main node/.style={circle,draw,font=\bfseries,minimum
		size=7mm}]

	\node[main node] (A1)  at ({-360/5*(1-1)}:-2.5)               {$a_1$};
	\node[main node] (A2)  at ({-360/5*(2-1)}:-2.5)                 {$a_2$};
	\node[main node] (A3)  at ({-360/5*(3-1)}:-2.5)
	{$a_3$};
	\node[main node] (A4)  at ({-360/5*(4-1)}:-2.5)
	{$a_4$};
	\node[main node] (A5)  at ({-360/5*(5-1)}:-2.5)
	{$a_5$};
	\path
	(A1) edge [loop below] node {} (A1)
	edge [bend left] node {} (A2)
	(A2) edge node {} (A1)
	(A3) edge node {} (A1)
	edge node {} (A2)
	(A4) edge node {} (A1)
	edge [loop right] node {} (A4)
	(A5) edge node {} (A1)
	edge node {} (A2)
	edge node {} (A3)
	edge node {} (A4);
	\end{tikzpicture} \caption{Qualification graph of example instance from
	Subsection
	\ref{sub:BV}.}\label{fi:qg}
\end{figure}
\paragraph{Example.} Let us consider an instance of the group identification
problem consisting of five
agents~$A=\{a_1,a_2,a_3,a_4,a_5\}$ and the following qualification profile
$\varphi$ as an example.
\[ \varphi=
\begin{blockarray}{cccccc}
& a_1 & a_2 & a_3 & a_4 & a_5
\\
\begin{block}{c(ccccc)}
a_1 &1 & 1 & 0 & 0 & 0  \\
a_2 &1 & 0 & 0 & 0 & 0 \\
a_3 &1 & 1 & 0 & 0 & 0  \\
a_4 &1 & 0 & 0 & 1 & 0  \\
a_5 &1 & 1 & 1 & 1 & 0  \\
\end{block}
\end{blockarray}
\]

We also include the qualification graph of this instance in \autoref{fi:qg}.
For the liberal-start-respecting rule \flsr, we start with the set of agents
$K_1=\{a_1,a_4\}$ that qualify themselves and have~$K_2=K_3=\{a_1,a_2,a_4\}$.
Thus, it holds that
$\flsr(A,\varphi)=\{a_1,a_2,a_4\}$. For the consensus-start-respecting
rule \fcsr, we start with the set of agents $K_1=\{a_1\}$ that are qualified by
everyone and have~$K_2=K_3=\{a_1,a_2\}$. Thus, it holds that
$\fcsr(A,\varphi)=\{a_1,a_2\}$. For the consent rule \fst, the set of socially
qualified agents depends on the chosen parameters $s$ and $t$. Let $s=t=3$,
i.e., an agent needs two other agents agreeing on its opinion about itself
to be classified according to its opinion. Thus, it holds that
$\fst(A,\varphi)=\{a_1,a_2\}$.

To exemplify the bribery problems we study in this paper, we present some
example questions and solutions for the above instance. Let us start with the
liberal-start-respecting rule \flsr and \textsc{Link Bribery}. Assuming that
we are in the
\textsc{Constructive} setting and the goal is to make all agents socially
qualified ($A^+=A$), there
exist four optimal solutions, i.e., bribe one of $a_1$, $a_2$, $a_4$, and
$a_5$ to
qualify $a_5$.
Assuming that we are in the \textsc{Constructive}+\textsc{Destructive} setting
and the goal is to make $a_5$
socially qualified and $a_3$ socially disqualified ($A^+=\{a_5\}$ and
$A^-=\{a_3\}$), one optimal solution is to bribe $a_5$ to qualify itself and
to disqualify $a_3$.
We now turn to the consent rule \fst with $s=t=3$ and \textsc{Agent Bribery}.
Assuming
that we are in the \textsc{Destructive} setting and the goal is to make
everyone disqualified ($A^-=A$), one optimal solution is to bribe $a_1$,
$a_2$, and $a_5$ to disqualify everyone. Finally, assuming that we are in the
\textsc{Exact} setting and the goal is to make~$a_1$ and $a_4$ socially
qualified and everyone else socially disqualified ($A^+=\{a_1,a_4\}$), the
optimal solution is to either bribe $a_1$ or $a_3$ to disqualify $a_2$ and
qualify $a_4$.

\subsection{Graph Theory, Graph Algorithms, and Parameterized Complexity}
\label{sub:GTPC}

A graph is a pair $(V,E)$ of \emph{vertices} and \emph{edges}; if the graph is directed, then $E$ is the set of \emph{arcs}.
A \emph{walk} from $s$ to $t$ in a (directed) graph is a sequence $v_0 e_1 v_1
\ldots e_p v_p$ with $s = v_0$ and $t = v_p$, where $e_i = (v_{i-1},v_i)$ for
all $i = 1, \ldots, p$.
A path is a walk where each vertex occurs at most once.
Let $G = (V,E)$ be a graph, $F \subseteq E$ a set of edges, and $U \subseteq V$ a set of vertices.
We write $G \setminus F$ to denote the graph $(V, E \setminus F)$ and
$G-U$ to denote the graph with vertex set $V \setminus U$ whose edge set is the set of
edges $e \in E$ with both endvertices in $V \setminus U$.
We use standard graph notation; see, e.g., the monograph
of~\cite{matousek-nesetril}. We now define several graph problems that we use
in the construction of different algorithms each accompanied by a statement
concerning the running time for solving the respective problem.

\defProblemTask{\textsc{Minimum Weighted Cut}}
{A directed graph~$G=(V,E)$ with arc weights~$w \colon V \times V \rightarrow \mathbb{N}$ and two distinct vertices $\sigma$ and $\tau$.}
{Find a set of arcs $F \subseteq E$ minimizing $\sum_{e \in F} w(e)$ such that there is no path from $\sigma$ to $\tau$ in $G \setminus F$.}

\begin{fact}[{\citep{Orl13}}]\label{prop:solving:minimumWeightCut}
  \textsc{Minimum Weighted Cut} can be solved in \mbox{$O(|V| \cdot |E|)$} time.
\end{fact}

\defProblemTask{\textsc{Minimum Weighted Separator}}
{A directed graph~$G=(V,E)$ with vertex weights~$w \colon V \rightarrow \mathbb{N}$ and two distinct vertices $\sigma$ and $\tau$.}
{Find a set of vertices $U \subseteq V \setminus \{ \sigma, \tau\}$ minimizing $\sum_{u \in U} w(u)$ such that there is no path from $\sigma$ to $\tau$ in $G - U$.}

One can solve \textsc{Minimum Weighted Separator} using \textsc{Minimum
Weighted Cut} due to, e.g., the following folklore construction: Given a graph
$G=(V,E)$, to construct the corresponding flow network, every
vertex~$x\in V$ is replaced by two vertices~$x_{\opin}$ and~$x_{\opout}$ such
that every ingoing arc to~$x$
is replaced by an ingoing arc to~$x_{\opin}$ and every outgoing arc from~$x$ is replaced by an
outgoing arc from~$x_{\opout}$.
Furthermore, one adds an arc with weight~$w(x)$ from~$x_{\opin}$
to~$x_{\opout}$ and sets
the weight of all other arcs to~$\infty$.
It is easy to verify that now a minimum weighted cut corresponds to an minimum weighted $(\sigma, \tau)$-separator
of the original graph.

\begin{fact}[{\citep{Orl13}}]\label{prop:solving:minimumWeightSeparator}
  \textsc{Minimum Weighted Separator} can be solved in \mbox{$O(|V| \cdot |E|)$} time.
\end{fact}

\defProblemTask{\textsc{Minimum Weighted Spanning Arborescence}}
{A directed graph~$G=(V,E)$ with arc weights~$w \colon V \times V \rightarrow \mathbb{N}$ and a root~$r$.}
{Find a set of arcs~$F \subseteq E$ of the minimum weight such that in the graph $(V,F)$ every vertex is reachable by a unique directed path from the root~$r$.}

\begin{fact}[{\citep{Edm67,Chu65,GGST86}}]\label{prop:solving:minimumWeightSpanningArborescence}
  \textsc{Minimum Weighted Spanning Arborescence} can be solved in \mbox{$O(|E| + |V| \log |V|)$} time.
\end{fact}

\defProblemTask{\textsc{Weighted Directer Steiner Tree} (WDST)}
{A directed graph~$G=(V,E)$ with arc weights~$w \colon V \times V \rightarrow \mathbb{N}$, a set of terminals~$T \subseteq V$, and a root vertex~$s$.}
{Find a subset of arcs of~$G$ minimizing the total weight such that there is a directed path from the root~$s$ to every terminal~$t \in T$.}

\begin{fact}[{\citep{Ned09}}]\label{prop:solving:weightedDirectedSteinerTree}
  There is an algorithm solving \textsc{Weighted Directer Steiner Tree} whose running-time depends exponentially on $|T|$ (in fact, by \mbox{$2^{|T|}$}) and polynomially on the size of the input and the maximum arc weight.
\end{fact}

\paragraph{Parameterized Complexity.}

To provide a fine-grained computational complexity analysis of our problems,
we use tools from parameterized complexity~\citep{CFKLMPPS15,DF13,FG06,Nie06}.
Herein, one identifies a \emph{parameter}~$k$ (a positive integer)
and takes a closer look at the computational complexity of the problem with
respect to the parameter $k$ and input size.
A problem parameterized by~$k$ is called \emph{fixed-parameter tractable}
(in FPT) if it can be solved in $f(k) \cdot |I|^{O(1)}$ time, where $|I|$
is the size of a given instance, $k$ is the parameter value, and~$f$ is a
computable (typically super-polynomial) function.
In order to disprove fixed-parameter tractability, we use a well-known
complexity hierarchy of classes of parameterized problems:
$$\text{FPT} \subseteq \text{W[1]} \subseteq \text{W[2]} \subseteq \cdots \subseteq \text{XP}.$$
All these inclusions are widely believed to be proper.
Hardness for parameterized classes are defined through parameterized reductions
which are similar to classical polynomial-time many-one reductions.
In this paper, it suffices to use polynomial-time many-one reduction which
additionally ensure that the value of the parameter
in the problem we reduce to depends only on the value of the parameter of the
problem we reduce from.
To look for tractable cases beyond single parameters, we also consider
\emph{parameter combination}s.
For example, when referring to the combined parameter $k' + k''+k'''$,
we implicitly define a parameter $k = k' + k'' + k'''$.
Naturally, hardness for a combined parameter implies hardness for each single
parameter, while fixed-parameter tractability for a single parameter implies
fixed-parameter tractability for each parameter combination that involves
the respective parameter.

\section{Iterative Rules} \label{se:itrules}

We start by considering agent bribery, before we examine link bribery in the
subsequent subsection.
\subsection{Agent Bribery} \label{sub:itab}
For agent bribery, the already known positive results for constructive bribery
and destructive bribery \citep{DBLP:journals/aamas/ErdelyiRY20} extend to
constructive+destructive bribery and even to its priced version.
The algorithm for the general problem, however, is significantly more involved
than the known
ones, as both ``positive'' and ``negative'' constraints need to be simultaneously taken into
account.
The general idea of the algorithm is to bribe the agents that form a minimum
(weighted)
separator between the
agents in $A^+$ and the agents in $A^-$ in the (slightly altered)
qualification graph such that they qualify themselves and all agents in
$A^+$.

Before we start with formally proving this result,
we make an observation about
the start of the social qualification process of our iterative rules that allows
to simplify our algorithms significantly.
Note that for both iterative rules there are two reasons for an agent~$a$
to be socially qualified:
The first reason is that the agent~$a$ becomes socially qualified in the first phase,
i.e.,\ $a \in K_1$, because $a$~is either qualified by all other agents or
qualifies itself.
The second reason is that the agent~$a$ is qualified by one other agent~$a'$ that
was already qualified in a previous iteration.
For agent bribery, it is clear that one may have to use the first option for some instances
(if $K_1=\emptyset$ and $A^+\neq\emptyset$).
The second option, however, is intuitively much cheaper so that using the first option
should be necessary at most once.
That this is indeed the case is not immediately obvious, so that we show it formally
in the following observation.

\begin{observation}
 \label{obs:bribe1initial}
 For both \fcsr~and \flsr, if there is a successful agent bribery that makes
two agents~$a^*_1$ and~$a^*_2$ being initially socially qualified
 ($\{a^*_1,a^*_2\} \subseteq K_1$ after the bribery but $\{a^*_1,a^*_2\} \cap K_1 = \emptyset$ before the bribery),
 then there is also a successful bribery of the same cost that makes just
one agent from~$\{a^*_1,a^*_2\}$
 initially socially qualified ($|\{a^*_1,a^*_2\} \cap K_1|=1$ after bribery).
\end{observation}

\begin{proof}
 To obtain the bribery that makes just one agent from~$\{a^*_1,a^*_2\}$ initially socially qualified, we do the following.
 For \flsr, since $\{a^*_1,a^*_2\} \subseteq K_1$ after the bribery but 
 $\{a^*_1,a^*_2\} \cap K_1 = \emptyset$ before the bribery, we must bribe both 
 $a^*_1$ and $a^*_2$ in the original bribery.
 To make just $a^*_1$ initially socially qualified, we can bribe the same as 
 before 
 with the following adjustments: We make  
 $a^*_1$ to qualify $a^*_2$ and we make $a^*_2$ to disqualify itself.
 It is easy to see that the new bribery bribes the same set of agents and after 
 the new bribery $a^*_2 \in K_2 \setminus K_1$, as $a^*_2$ does not qualify 
 itself but it is qualified by $a^*_1$.
 
 For \fcsr, if one of the two agents, say without loss of generality~$a^*_2$, 
 is bribed in the original bribery, then bribe exactly in the same way as 
 before but~$a^*_2$ is now not qualifying itself anymore.
 Otherwise, select an arbitrary bribed agent (which is by assumption different 
 from $a^*_1$ and $a^*_2$) that did not originally qualify~$a^*_2$ and keep it 
 disqualifying~$a^*_2$.
 In both cases, we bribe the same set of agents as in the original bribery.
 While with the new bribery agent~$a^*_2$ does not become initially socially 
 qualified ($a^*_2 \notin K_1$), $a^*_2$ is still
 qualified by agent~$a^*_1$ and thus becomes socially qualified in the second 
 iteration ($a^*_2 \in K_2$).
 It is easy to verify that the final outcome with respect to the agent's social
qualification also remains the same for all other agents (though, the 
qualification process might be different).
\end{proof}

By the above observation, it is never necessary to make more than one agent initially socially qualified,
which we will use in our algorithms.

\begin{theorem}
	\label{thm:fsr-cd-agent}
	\mixed \pAgentBribery{\flsr} is solvable in time~$O(n^3)$, while
	\mixed \pAgentBribery{\fcsr} is solvable in time~$O(n^4)$.
\end{theorem}
\begin{proof}
  We first provide the algorithms for the two rules and then, since both are based on finding a minimum weighted separator, we analyze their running time.

	 \medskip\noindent\textbf{LSR:} Let us focus on~\flsr~first.
	Let~$L=\{a \in A^- \mid \varphi(a,a)=1\}$ be the set of all agents in~$A^-$ who qualify themselves.
	We first bribe all agents~$a \in L$ such that they disqualify everyone.
	To determine which further agents we want to bribe, we try to find a minimum weighted separator between the vertices $A^+$ and~$A^-$
	in an auxiliary graph based on the qualification graph.
	More precisely, we add a source vertex~$\sigma$ to the qualification graph
	and connect~$\sigma$ to all vertices in~$A^{+}$ and vertices with a self-loop.
	Moreover, we merge all vertices in $A^-$ into one sink vertex $\tau$ (see
	\autoref{fig:fsr-cd-agent} for a visualization).
	We assign the weight~$\rho(a)$ to each vertex~$a\in A\setminus A^{-}$.
	We call the resulting graph~$G$.
	Subsequently, we calculate a minimum weighted $(\sigma, \tau)$-separator
	$A^{'}$ in $G$:

	   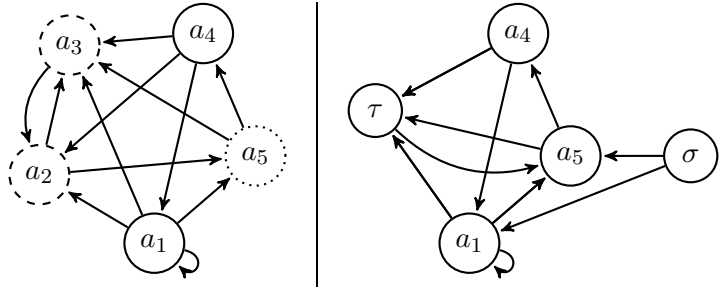
\begin{figure}[bt]
		\begin{center}
			\begin{tikzpicture}[rotate=95,->,>=stealth',scale=0.6,shorten
			>=1pt,auto,node
			distance=4cm,
			thick,main node/.style={circle,draw,font=\bfseries,minimum
			size=7mm}]
			\node[main node] (A)  at
			({-360/5*(1-1)}:-2.5) {$a_1$};
			\node[main node, dashed] (B)  at ({-360/5*(2-1)}:-2.5) {$a_2$};
			\node[main node, dashed] (C)  at
			({-360/5*(3-1)}:-2.5)
			{$a_3$};
			\node[main node] (D)  at ({-360/5*(4-1)}:-2.5)
			{$a_4$};
			\node[main node, dotted] (E)  at
			({-360/5*(5-1)}:-2.5)
			{$a_5$};
			\path
			(A) edge [loop below] node {} (A)
			edge node {} (B)
			edge node {} (C)
			edge node {} (E)
			(B) edge node {} (C)
			edge node {} (E)
			(C) edge [bend right] node {} (B)
			(D) edge node {} (A)
			edge node {} (B)
			edge node {} (C)
			(E) edge node {} (C)
			edge node {} (D);
			\end{tikzpicture}
			\hspace{0.3cm}\vline\hspace{0.3cm}
			\begin{tikzpicture}[rotate=95,->,>=stealth',scale=0.6,shorten
			>=1pt,auto,node
			distance=4cm,
			thick,main node/.style={circle,draw,font=\bfseries,minimum
			size=7mm}]
			\node[main node] (A)  at
			({-360/5*(1-1)}:-2.5) {$a_1$};
			\node[main node] (S)  at (0.6,1.8) {$\tau$};
			\node[main node] (D)  at ({-360/5*(4-1)}:-2.5)
			{$a_4$};
			\node[main node] (E)  at ({-360/5*(5-1)}:-2.5)
			{$a_5$};
			\node[main node] (SS)  at (-1,-5) {$\sigma$};
			\path
			(A) edge [loop below] node {} (A)
			edge node {} (S)
			edge node {} (S)
			edge node {} (E)
			edge node {} (E);
			\path (S) [bend right] edge node {} (E);
			\path (D) edge node {} (A)
			edge node {} (S)
			edge node {} (S)
			(E) edge node {} (S)
			edge node {} (D)
			(SS) edge node {} (A)
			edge node {} (E);
			\end{tikzpicture}
		\end{center}
		\caption{Visualization of the transformation of a qualification graph
			where agents from $A^+$ are drawn as dotted circles and agents from
			$A^-$
			as dashed circles from
			\autoref{thm:fsr-cd-agent} for \flsr.
		}
		\label{fig:fsr-cd-agent}
	\end{figure}

	If~$A^{'}=\emptyset$ and all agents~$a \in A^+$ are already socially qualified, we are done. If~$A^{'}=\emptyset$ and this is not the case, we bribe an agent $a\in A \setminus A^-$ of the minimum weight and make it qualify itself and all agents in~$A^+$ and disqualify all other agents.
	Now, all agents in~$A^+$ are socially qualified, and since~$A^{'}=\emptyset$, all agents in~$A^-$ are socially disqualified.

	If~$A^{'} \not = \emptyset$, as the existence of a $(\sigma,\tau)$-path in 
	the constructed graph implies that some agent in~$A^-$ will be socially 
	qualified as soon as we make all agents in~$A^+$ socially qualified, we 
	need to bribe agents of total weight  at least $\sum_{a\in A^{'}}\rho(a)$ 
	to 
	destroy all $(\sigma,\tau)$-paths.

	In fact, we bribe all agents in $A^{'}$ such that they qualify themselves and all agents in $A^{+}$ and disqualify all other agents.
	Let~$G'$ be the graph which is obtained from $G$ by deleting all arcs
	corresponding to the qualifications that got deleted and adding arcs
	corresponding to the qualifications that got inserted. By construction, there
	is no~$(\sigma, \tau )$-path in $G'$. Moreover, in the modified 
qualification
	profile, all agents in $A^{+}$ are socially qualified, as they are
	qualified by at least one agent $a\in A'$ qualifying itself.
	We claim that no agent in $A^{-}$ can be socially qualified.
	Assume that there exists an agent $a^{-}\in A^{-}$ that is socially qualified.
	Then, there exists a path from some agent~$a^*$ which qualifies itself to $a^-$ in the altered instance.
	If~$a^* \not \in A^{'}$, then such a path implies that there exists an $(\sigma, \tau )$-path in~$G'$ which cannot be the case.
	If~$a^* \in A^{'}$, since~$a^*$ only qualifies agents from~$A^{+}$ and itself, there is also a path from some agent in~$A^{+}$ to $a^-$ in the altered instance.
	This also implies that there needs to exist an $(\sigma, \tau )$-path
	in~$G'$ which again cannot be the case.

	\medskip\noindent\textbf{CSR:} For \fcsr, first assume
that~$A^{+}=\emptyset$.
  If no agent from $A^-$ is socially qualified, then we are done.
  Otherwise, we pick an arbitrary minimum-weight agent and make it disqualify everyone.

	In the following, we assume that~$A^{+} \not = \emptyset$.
 	If~$A^*=\emptyset$, since~$A^{+} \not = \emptyset$, there has to be at
 	least one agent who is qualified by everyone after the bribery.
 	We guess an agent $a^*$ that should be qualified by everyone in the end
 	(by \autoref{obs:bribe1initial} one such agent will be enough).
 	For each guess~$a^*$, we bribe all agents that do not already qualify $a^*$
 	such that they qualify all agents in~$A^+ \cup \{a^*\}$ and disqualify all
 	other agents.
 	We still denote the set of agents who are qualified by everyone after the
 	above bribery as~$A^*$.

 	Then similar to the algorithm for \flsr, we insert a source vertex
 	$\sigma$ in the	underlying qualification graph and connect it to all
 	vertices in $A^+ \cup A^*$.
 	Moreover, we contract all vertices from $A^-$ into a sink vertex~$\tau$.
 	As for \flsr, we assign the weight~$\rho(a)$ to each vertex~$a\in A\setminus A^{-}$.
 	We remove self-loops and call the resulting auxiliary graph~$G$.
 	Subsequently, we calculate a minimum weighted $(\sigma, \tau)$-separator $A'$
 	in $G$ and bribe all
 	agents in $A'$ to qualify all agents in $A^+ \cup A^*$ and to disqualify
 	all other agents. Subsequently, if not all agents in $A^+$ are socially
qualified, we additionally bribe a cheapest agent~$\widetilde{a}$ that is
socially qualified after the already described bribery and let it qualify
all agents in $A^+ \cup A^*$ and disqualify
 	all other agents (note that such an agent needs to exist, as $a^*$ is
socially qualified).

    First of all, note that the described bribery is always successful.
  	In the case where no agent~$\widetilde{a}$ in addition to the agents from
  	$A'$ is bribed,
  	the presented bribery is also obviously optimal, as it is always necessary
  	to
  	bribe agents at a cost at least~$\rho(A')$ to socially disqualify all agents
    from~$A^-$ (assuming that $a^*$ was guessed correctly). Therefore it remains to
consider the case where an agent $\widetilde{a}$ in addition to the agents $A'$
was bribed. Then, no agent from the separator $A'$ is socially
qualified before the bribery of $\widetilde{a}$. This in turn implies that the
	$(\sigma, \tau)$-separator $A'$ touches no path
	between a vertex from~$A^*$ and~$\tau$, because there never was such a path.
	From the absence of a path between $A^*$ and $\tau$ it follows that no
agent that is part of a minimal (and also not minimum weighted) $(\sigma,
\tau)$-separator can be socially qualified before the bribery of
$\widetilde{a}$. Thus, independent of which minimal separator we choose
and how we bribe the agents in the separator, the set of agents being socially
qualified after the bribery of these agents remains the
same. Thus, we need at least cost $\rho(A')$ to separate the agents from $A^+$
and $A^-$ in this case and at least cost $\rho(\widetilde{a})$ to make all
agents from $A^+$ socially qualified. This minimum cost is exactly the cost of
the described bribery.

 	Finally, to bound the running time, observe that each auxiliary graph~$G$
 	has $O(|A|)$~nodes and $O(|A|^2)$~arcs.
 	Finding a minimum weighted $(\sigma, \tau)$-separator can thus be done in
 	$O(|A|^3)$ time by Fact~\ref{prop:solving:minimumWeightSeparator}.
 	For \fcsr, we get an additional factor~$|A|$, since we may have to guess 
 	the initially qualified agent~$a^*$.
\end{proof}

From this it immediately follows that (priced) agent bribery is also
polynomial-time
solvable for all other considered variants including the previously
completely unstudied
case of exact bribery.
\begin{corollary} \label{co:flsr-e-arc}
	\exact/\const/\dest \pAgentBribery{\fcsr/\flsr} is solvable in polynomial time.
\end{corollary}

\subsection{Link Bribery} \label{sub:itlb}
We now turn to the setting of link bribery and settle the complexity of the
related decision problem for all bribery goals and both iterative rules. We
start by proving that the problem is polynomial-time solvable for
(priced) destructive
bribery and (priced) exact bribery. For destructive bribery, similar to
Theorem~\ref{thm:fsr-cd-agent}, we separate the agents that are initially
socially qualified from the agents in $A^-$ in the qualification graph.
However, here, as we pay per changed qualification, we need to calculate a
minimum weighted cut. For exact bribery, we need to separate the agents from
$A^-$ and $A^+$ while making sure that all agents from $A^+$ are socially
qualified in the end.
\begin{theorem}
	\label{thm:flsr-d-arc}
	\dest \pLinkBribery{\flsr} can be solved in~$O(n^3)$ time,
	\dest \pLinkBribery{\fcsr} can be solved in~$O(n^4)$ time,
	\exact \pLinkBribery{\flsr} can be solved in~$O(n^2)$ time, and
	\exact \pLinkBribery{\fcsr} can be solved in~$O(n^3)$ time.
\end{theorem}
\begin{proof}
	We first prove the theorem for \dest \pLinkBribery{\flsr}.

	\medskip\noindent\textbf{Destructive LSR:} We start by bribing all agents
	from $A^-$ who qualify themselves such that they all disqualify themselves
	and substract the corresponding costs from the budget~$\ell$ and update
	the profile~$\varphi$.
	Let $G$~be the qualification graph of the altered instance.
	We now consider a slightly modified version of~$G$ and calculate a minimum
	weighted cut to solve the problem.
	We start with~$G$ and assign to every arc~$e$ the weight~$\rho(e)$.
	Subsequently, we add a source vertex~$\sigma$~to $G$ and
	connect~$\sigma$ to
	each vertex~$a$ with a self-loop by an arc of weight~$\rho((a,a))$;
	afterwards, self-loops can be removed.
	Additionally, we introduce a sink vertex~$\tau$ and connect every vertex
	from~$A^-$ to the sink $\tau$ using arcs with weight~$\ell+1$.
	Now, we compute a minimum weighted $(\sigma,\tau)$-cut $E'\subseteq E$
(using Fact~\ref{prop:solving:minimumWeightCut}).
If the weight of a minimum weighted cut is at least~$\ell+1$, then we
reject the instance.
 	If the weight of a minimum weighted $(\sigma,\tau)$-cut is at
most~$\ell$, then it cannot
 	contain any arc that ends in $\tau$
	and we remove the qualifications corresponding to the cut from the
qualification profile:
	If $(a,a')\in E'$ for some $a,a'\in A$, we make $a$ disqualify $a'$.
	If $(\sigma,a)\in E'$ for some~$a\in A$, we make~$a$ disqualify itself.

	On the one hand, after this bribery, no agent $a^-\in A^-$ is socially qualified,
	as the absence of a $(\sigma,\tau)$-path implies that there does not exist a path
	from an agent qualifying itself to $a^-$ in the altered qualification profile.
	On the other hand, the described bribery is optimal, as the existence of a
	$(\sigma,\tau)$-path always implies that at least one agent in $A^-$ is socially
	qualified in the end.

	\medskip\noindent\textbf{Destructive CSR:} For \dest \pLinkBribery{\fcsr},
we need  to follow a slightly more involved
 	approach.
 	For some $a\in A$, we denote by~$\rho^{\cheap}(a)$ the cheapest possible
 	cost of bribing some
 	agent~$a'$ currently qualifying $a$ to disqualify~$a$,
 	i.e. $\rho^{\cheap}(a)=\min_{a''\in Q^+(a)} \rho(a'',a)$, and fix an agent
 	$\cheap(a) \in \argmin_{a''\in Q^+(a)} \rho(a'',a)$.
 	Let $d=\sum_{a \in A^*} \rho^{\cheap}(a)$ be the cheapest costs of bribing
 	for each agent~$a$ from the set~$A^*$ of agents that are qualified by everyone
 	at least one agent~$a'$ to disqualify~$a$.
 	(Note that, since we pay per modified qualification, costs for bribing the same
 	agent twice, i.e.\ when $\cheap(a)=\cheap(a')$ for some $a\neq a'$, are indeed
 	to be summed up. That is, the agents are handled independently from each other.)
 	If $d\leq \ell$, we accept as we can make all agents from $A^*$ socially
 	disqualified
 	(and thus no agent and, in particular, no agent from~$A^-$ is socially
 	qualified) by bribing~$\{(\cheap(a),a) \mid a \in A^*\}$.
 	In the following, we assume $\ell<d$.

 	As $\ell<d$, after the bribery there will be at least one agent from~$A^*$ that remains to
 	be qualified by all agents.
 	We guess this agent and denote it as~$a^*$ in the following.
 	Observe that if we have correctly guessed $a^*$, then
 	an optimal bribery will never bribe agents from~$A^*\setminus\{a^*\}$ to disqualify
 	themselves.
 	The only effect bribing an agent~$\hat{a}$ from~$A^*\setminus  \{a^*\}$ to
disqualify
 	itself can have is that~$\hat{a}$ may not be initially socially qualified
anymore.
 	However, agent~$a^*$ (assumed to be guessed correctly as an agent that remains initially
 	socially qualified) anyway qualifies all agents from~$A^*$ including~$\hat{a}$, except if we
 	bribe~$a^*$ to disqualify~$\hat{a}$.
 	If this, however, is the case, then bribing~$\hat{a}$ to disqualify itself was useless.

 	We can now compute a minimum weighted cut after modifying the qualification graph~$G$ as follows.
 	Again, we start with~$G$ and assign to every arc~$e$ the weight~$\rho(e)$.
 	We add a source vertex~$\sigma$ and a sink vertex~$\tau$.
 	We add an arc with weight~$\ell+1$ from~$\sigma$ to $a^*$.
 	We connect every vertex from~$A^-$ by an arc of weight~$\ell+1$ to~$\tau$.
 	Now, we compute a minimum weighted $(\sigma,\tau)$-cut in the modified
qualification graph.
 	If the weight of a minimum weighted cut is at least~$\ell+1$, then we
reject the current guess of~$a^*$,
 	because from this it follows that any bribery that keeps~$a^*$ initially
socially qualified is too expensive.
 	If the weight of a minimum weighted $(\sigma,\tau)$-cut is at most~$\ell$,
then it cannot
 	contain any arc that touches~$\sigma$ or~$\tau$.
 	Hence, it is possible to delete all qualifications corresponding to arcs in the
 	minimum cut to obtain a minimal bribery.
 	Similar to above, the described bribery is successful and minimal assuming~$a^*$ was
 	correctly guessed, as there exists  an $(\sigma,\tau)$-path in the
constructed graph
 	if and only if an agent from~$A^-$ is socially qualified.

 	Our algorithm keeps track and returns a cheapest bribery over all guesses of~$a^*$ (if any)
 	and returns ``no'' if all guesses are rejected.
 	It is easy to verify that in the latter case indeed no bribery at cost at
 	most~$\ell$ can exist.

	\medskip
 	We now turn to \exact \pLinkBribery{\flsr/\fcsr}.
 	
 	\noindent\textbf{Exact LSR:} For \flsr, we start by making all
agents~$A^-$ disqualify themselves if this
 	is not already the case. Subsequently, we remove all qualifications from an
 	agent in~$A^+$ to an agent in~$A^-$.

 		   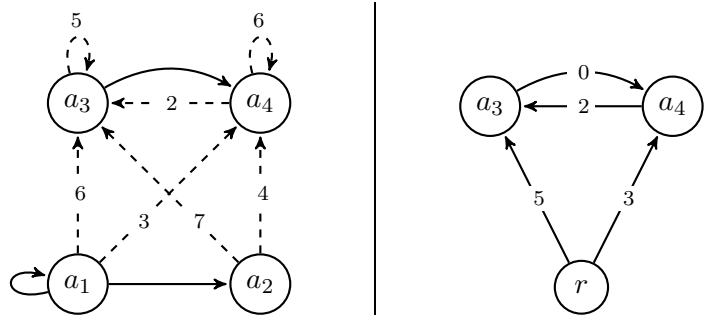
\begin{figure}[bt]
 		\begin{center}
 			\begin{tikzpicture}[->,>=stealth',scale=0.6,shorten
 			>=1pt,node
 			distance=4cm,
 			thick,main node/.style={circle,draw,font=\bfseries,minimum
 				size=7mm}]
 			\node[main node] (A)  at
 			(0,0) {$a_1$};
 			\node[main node] (B)  at
 			(4,0) {$a_2$};
 			\node[main node] (C)  at
 			(0,4) {$a_3$};
 			\node[main node] (D)  at
 			(4,4) {$a_4$};
 			\path
 			(A) edge [loop left] node[] {} (A)
 			edge node[]
 			{} (B)
 			edge [dashed] node[pos=0.5, fill=white, inner sep=2pt]
 			{\scriptsize~$6$} (C)
 			edge [dashed] node[pos=0.3, fill=white, inner sep=2pt]
 			{\scriptsize~$3$} (D);
 			\path
 			(B) edge [dashed] node[pos=0.3, fill=white, inner sep=2pt]
 			{\scriptsize~$7$} (C)
 			edge [dashed] node[pos=0.5, fill=white, inner sep=2pt]
 			{\scriptsize~$4$} (D);
 			\path
 			(C) edge [dashed, loop above] node[pos=0.4, inner
 			sep=4pt]
 			{\scriptsize~$5$} (C)
 			edge [bend left]  (D);
 			\path
 			(D) edge [dashed,loop above] node[pos=0.4, inner
 			sep=4pt]
 			{\scriptsize~$6$} (D)
 			edge [dashed] node[pos=0.5, fill=white, inner sep=2pt]
 			{\scriptsize~$2$} (C);
 			\end{tikzpicture}
 			\hspace{1cm}\vline\hspace{1cm}
 			\begin{tikzpicture}[->,>=stealth',scale=0.6,shorten
 			>=1pt,node
 			distance=4cm,
 			thick,main node/.style={circle,draw,font=\bfseries,minimum
 				size=7mm}]
 			\node[main node] (A)  at
 			(2,0) {$r$};
 			\node[main node] (C)  at
 			(0,4) {$a_3$};
 			\node[main node] (D)  at
 			(4,4) {$a_4$};
 			\path
 			(A)	edge node[pos=0.5, fill=white, inner sep=2pt]
 			{\scriptsize~$5$} (C)
 			edge node[pos=0.5, fill=white, inner sep=2pt]
 			{\scriptsize~$3$} (D);
 			\path
 			(C) edge [bend left] node[pos=0.5, fill=white, inner sep=2pt]
 			{\scriptsize~$0$}  (D);
 			\path
 			(D) edge node[pos=0.5, fill=white, inner sep=2pt]
 			{\scriptsize~$2$} (C);
 			\end{tikzpicture}
 		\end{center}
 		\caption{Visualization of the transformation of an instance with
 		$A=\{a_1,a_2,a_3,a_4\}$ and $A^+=A$ of \exact \pLinkBribery{\flsr} from
 		\autoref{thm:flsr-d-arc}. On the left side, the solid arcs
 		display
 		the qualification graph, while the dashed arcs indicate the cost of
 		adding the corresponding qualification. On
 		the right
 		side, the corresponding weighted graph constructed as described in
 		\autoref{thm:flsr-d-arc} is displayed. The instance admits two
 		solutions, that is, bribing $a_3$ to qualify itself or bribing $a_4$ to
 		qualify $a_3$ and $a_1$ to qualify $a_4$.
 		}\label{fi:th2}
 	\end{figure}

 	To ensure that
 	all agents in $A^+$ are socially qualified, we consider again
 	the underlying qualification graph restricted to the vertices from~$A^+$
 	and slightly modify it as follows to solve the problem by finding a
spanning arborescence of the
 	minimum weight. For a visualization of the following construction see
 	\autoref{fi:th2}.
 	Let $R$~be the set of vertices which have a self-loop or which are reachable
 	from a vertex with a self-loop.
 	We merge all vertices from~$R$ into a single new vertex~$r$;
 	when $R=\emptyset$ this means to simply add a new vertex~$r$.
 	Let $a,a' \notin R$ be two vertices corresponding to not yet socially qualified agents.
 	If there is an arc from~$a$ to~$a'$ in the qualification graph,
 	then we set its weight to~$0$.
 	If there was no arc from~$a$ to~$a'$, then we create a new arc from~$a$
to~$a'$
 	and set its weight to~$\rho((a,a'))$ (the cost of bribing~$a$ to qualify~$a'$).
 	Furthermore, we create an arc from~$r$ to each vertex~$a' \notin R$
 	and set the weight to~$\min\left( \rho((a',a')), \min_{a \in R}
\rho((a,a')) \right)$
 	(the cheapest cost of either bribing $a'$ to qualify itself or
 	 of bribing an agent from~$R$ to qualify~$a'$).
 	What remains is to find a spanning arborescence of the minimum weight with root~$r$ using Fact~\ref{prop:solving:minimumWeightSpanningArborescence}.
 	It is easy to verify that a spanning arborescence of the minimum weight
 	corresponds to a cheapest possible bribery that makes all agents from~$A^+$
 	socially qualified and vice versa.

 	\medskip\noindent\textbf{Exact CSR:} For \fcsr, if $A^+$ is empty, we are
either done or there exist socially
 	qualified agents.
 	In the latter case, we denote by~$\rho^{\cheap}(a)$ the cheapest possible cost of bribing some
 	agent~$a'$ to disqualify~$a$, i.e., $\rho^{\cheap}(a)=\min_{a''\in Q^+(a)} \rho(a'',a)$,
 	and fix an agent $\cheap(a) \in \argmin_{a''\in Q^+(a)} \rho(a'',a)$.
 	By this definition, removing all qualifications from~$\{(\cheap(a),a) \mid
 	a \in A^*\}$ gives an optimal bribery.

 	Otherwise, if~$A^+$ is nonempty, we start by deleting all qualifications
 	from agents in~$A^+$ to agents in~$A^-$.
 	Subsequently, similar to the above, we try to make all agents in $A^+$
 	socially qualified by computing a spanning arborescence of the minimum weight
 	of a graph based on the qualification graph.
 	Handling the set of initially socially qualified agents, however, will be again more
 	complicated (similar as for \dest \fcsr \pLinkBribery{}).
 	Observe that if~$A^+$ is nonempty, then after the bribery the
 	set of initially socially qualified agents (which are qualified by all
 	agents)
 	cannot be empty.
 	If $A^*$ is nonempty without any bribery, then we continue very similar as
 	for \flsr:
 	Again, to ensure that all agents in $A^+$ are socially qualified, we consider
 	the underlying qualification graph restricted to the vertices from~$A^+$
 	and slightly modify it to solve the problem by finding a spanning
arborescence of the minimum weight.
 	Let~$R$ be the set of vertices which are initially qualified by all agents
or
 	which are reachable from a vertex initially qualified by all agents.
 	We merge all vertices from~$R$ into a single new vertex~$r$.
 	Let $a,a' \notin R$ be two vertices corresponding to not yet socially qualified agents.
 	If there is an arc from~$a$ to~$a'$ in the qualification graph,
 	then we set its weight to~$0$.
 	If there was no arc from~$a$ to~$a'$, we create a new arc from~$a$ to~$a'$
 	and set its weight to~$\rho((a,a'))$ (the cost of bribing~$a$ to qualify~$a'$).
 	Furthermore, we create an arc from~$r$ to each vertex~$a' \notin R$
 	and set the weight to~$\min_{a \in R} \rho((a,a'))$
 	(bribing an agent from~$R$ to qualify~$a'$).\footnote{Note that here,
 	opposed to the case with \flsr, it cannot be beneficial to make
 	an agent~$a$ qualifying itself and make it qualified by all agents so that
 	it would
 	be
 	part of the set of initially socially qualified agents because either the
 	agent is already
 	qualified by some agent from $R$ (which implies that $a$ is already
 	socially qualified) or this is at most as expensive
 	as making the cheapest agent from~$R$ qualifying~$a$.}
 	What remains is to find a spanning arborescence of the minimum weight with
root~$r$ (again by Fact~\ref{prop:solving:minimumWeightSpanningArborescence}).
 	It is easy to verify that a spanning arborescence of the minimum weight
 	corresponds to a cheapest possible bribery that makes all agents from~$A^+$
 	socially qualified and vice versa.

 	Finally, let~$A^*$ be empty (and~$A^+$ be still nonempty).
 	We continue by guessing an agent $a^*$ in $A^+$ that is qualified by everyone.
 	We make $a^*$ qualified by all agents and proceed as in the case where~$A^*$
 	was nonempty.

 	For the running time bounds, recall that the constructed graph will have
 	$O(n^2)$~edges and~$O(n)$~nodes.
 	Moreover, for~\fcsr, we get an additional factor~$n$ for guessing the
 	initially qualified agent.
 	Applying the known running time bounds for computing a minimum weighted cut or a
 	spanning arborescence of the minimum weight from
\autoref{prop:solving:minimumWeightCut} and
\autoref{prop:solving:minimumWeightSpanningArborescence} gives the bounds
claimed in the theorem.
\end{proof}

In contrast to this, the corresponding problem for constructive bribery is
NP-complete and even W[2]-hard parameterized by the budget $\ell$. This
difference in the complexity of the problem for constructive
bribery and destructive bribery is somewhat surprising, as their complexity is
the same in the case of agent bribery. We show the
hardness of \const \LinkBribery{\fcsr/\flsr} by a reduction from \textsc{Set
Cover}, which is NP-complete and W[2]-hard with respect to the requested size
of the cover. The general idea of the reduction is to introduce one agent for
each element (these form the set $A^+$) and for each set, where each set-agent
qualifies the agents corresponding to the elements in the set.
Notably, as all our hardness results, the hardness holds already for unit costs.
\begin{theorem}
	\label{thm:fsr-c-arc}
	\const \LinkBribery{\fcsr/\flsr} is NP-complete and W[2]-hard with respect to~$\ell$.
\end{theorem}

 \begin{proof}
 	We reduce from the \textsc{Set Cover} problem, where given a number $k$, a
 	universe of elements $U$, and a set of subsets of $\mathcal{S}\subseteq
 	2^{U}$ with $\bigcup\mathcal{S}=U$, one has to decide whether there exists
 	a subset $\mathcal{S'}\subseteq \mathcal{S}$ of size at most $k$ which
 	covers $U$.

 	We prove the theorem for {\fcsr} and {\flsr} by the same reduction from
 	\textsc{Set Cover}. Given an instance $(\mathcal{S},U,k)$ of \textsc{Set
 	Cover}, we construct an instance of \const \LinkBribery{\fcsr/\fcsr} with
 	budget $\ell=k$, agents $A=\mathcal{S}\cup U \cup a^{*}$ and $A^{+}=U$. The
 	qualification profile is created as follows: $a^{*}$ only qualifies
 	itself. Each $S \in \mathcal{S}$ qualifies all $a\in S$ and $a^{*}$ and
 	disqualifies all other agents. Each $a\in U$ qualifies $a^{*}$ and
 	disqualifies all other agents. Consequently, without any bribery $a^{*}$
 	and no one else is socially qualified for both \fcsr and \flsr.
 	We now prove that there exists a successful bribery of cost at most $\ell$
in the constructed instance if and only if there exists a set cover of size at
most $k$ in the given \textsc{Set Cover} instance.

 	$(\Rightarrow)$ Assume that $\mathcal{S}'$ is a cover of $U$ of size at
 	most $k$. Then, the briber bribes $a^{*}$ such that $a^{*}$ qualifies all
 	agents in $\mathcal{S}'$ which requires at most $\ell$ insertions. As
$\mathcal{S}'$ is a cover of $U$, all agents from $A^+$ are socially qualified
after the bribery.

 	$(\Leftarrow)$ Assuming that there exists a successful bribery resulting
in the qualification profile $\varphi'$, there also
 	exists a successful bribery where only $a^*$ is bribed and only links from
 	$a^*$ to some agent from $\mathcal{S}$ are inserted, as it is possible to
 	replace every inserted qualification $(a,a')$ for two $a,a' \in A$ with
 	$(a^{*},a')$ and every arc pointing to some vertex $u \in U$ by an arc from
 	$a^*$ to some vertex from $\mathcal{S}$ which contains $u$. Then, $\{S\in
 	\mathcal{S}\mid \varphi'(a^{*},S)=1 \}$ is clearly a cover of $U$ of size
at most $k=\ell$.
 \end{proof}

Apart from the parameter budget $\ell$, which might be small in most
applications as only agents which are close to the boundary of being socially
(dis)qualified might be interested in their precise margin, another natural
parameter is the set of agents we want to make socially qualified, i.e.,
$|A^+|$. This parameter may also be not too large in most applications, as one
may only be interested in the classification of a limited number of agents. In
contrast to the negative parameterized result for~$\ell$, by reducing the
problem to the \textsc{Weighted Directed Steiner Tree} (WDST) problem, it is
possible to prove that \const \pLinkBribery{\fcsr/\flsr} is FPT with
respect to
$|A^+|$.

 		   \begin{figure}[bt]
	\begin{center}
		\begin{tikzpicture}[->,>=stealth',scale=0.6,shorten
		>=1pt,node
		distance=4cm,
		thick,main node/.style={circle,draw,font=\bfseries,minimum
			size=7mm}]
		\node[main node] (A)  at
		(0,0) {$a_1$};
		\node[main node] (B)  at
		(4,0) {$a_2$};
		\node[main node] (C)  at
		(2,4) {$a_3$};
		\path
		(A) edge [loop left] node[] {} (A)
		edge [dashed] node[pos=0.5, fill=white, inner sep=2pt]
		{\scriptsize~$5$} (B)
		edge [dashed] node[pos=0.5, fill=white, inner sep=2pt]
		{\scriptsize~$5$} (C);
		\path
		(B) edge [dashed, loop right] node[pos=0.4, inner
		sep=2pt]
		{\scriptsize~$3$} (B)
		edge [bend left] node[] {} (A)
		edge [dashed] node[pos=0.3, fill=white, inner sep=2pt]
		{\scriptsize~$1$} (C);
		\path
		(C) edge [dashed, loop above] node[pos=0.4, inner
		sep=4pt]
		{\scriptsize~$4$} (C)
		edge [bend right]  (A)
		edge [bend left]  (B);
		\end{tikzpicture}
		\hspace{1cm}\vline\hspace{1cm}
		\begin{tikzpicture}[->,>=stealth',scale=0.6,shorten
		>=1pt,node
		distance=4cm,
		thick,main node/.style={circle,draw,font=\bfseries,minimum
			size=7mm}]
		\node[main node] (A)  at
		(0,0) {$a_1$};
		\node[main node] (B)  at
		(4,0) {$a_2$};
		\node[main node] (C)  at
		(2,4) {$a_3$};
		\node[main node] (D)  at
		(2,6.5) {$r$};
		\path
		(A) edge [] node[pos=0.5, fill=white, inner sep=2pt]
		{\scriptsize~$5$} (B)
		edge [] node[pos=0.5, fill=white, inner sep=2pt]
		{\scriptsize~$5$} (C);
		\path
		(B)	edge [bend left] node[pos=0.4, fill=white, inner
		sep=4pt]
		{\scriptsize~$0$} (A)
		edge [] node[pos=0.3, fill=white, inner sep=2pt]
		{\scriptsize~$1$} (C);
		\path
		(C) edge [bend right] node[pos=0.4, fill=white, inner
		sep=4pt]
		{\scriptsize~$0$}  (A)
		edge [bend left]  node[pos=0.4, fill=white, inner
		sep=4pt]
		{\scriptsize~$0$} (B);
		\path
		(D) edge [bend right] node[pos=0.4, fill=white, inner
		sep=4pt]
		{\scriptsize~$0$} (A)
		edge [bend left] node[pos=0.4, fill=white, inner
		sep=4pt]
		{\scriptsize~$3$} (B)
		edge [] node[pos=0.4, fill=white, inner
		sep=4pt]
		{\scriptsize~$4$} (C);
		\end{tikzpicture}
	\end{center}
	\caption{Visualization of the transformation of an instance with
		$A=\{a_1,a_2,a_3\}$ and $A^+=A$ of \const \pLinkBribery{\flsr} from
		\autoref{thm:fsr-c-fpt-arc}. On the left side, the solid arcs
		display
		the qualification graph, while the dashed arcs indicate the cost of
		adding the corresponding qualification. On
		the right
		side, the corresponding weighted graph constructed as described in
		\autoref{thm:fsr-c-fpt-arc} is displayed. The instance admits two
		solutions, that is, bribing $a_3$ to qualify itself or bribing $a_2$ to
		qualify itself and $a_2$ to qualify $a_3$.
	}\label{fi:th4}
\end{figure}

\begin{theorem}  \label{thm:fsr-c-fpt-arc}
	\const \pLinkBribery{\fcsr/\flsr} is FPT  with respect to $|A^+|$.
\end{theorem}
 \begin{proof}
 	We start by considering \flsr{} before we turn to \fcsr.

 	\medskip\noindent\textbf{LSR:} We examine a close
connection  between \const \pLinkBribery{\flsr} and WDST. For a visualization
of the following
construction see
 	\autoref{fi:th4}.
 	Let~$I$ be an instance of \const \pLinkBribery{\flsr} and
consider the qualification graph where we add one additional root vertex~$r$.
 	Moreover, for each existing arc set its weight to $0$ and replace every
 	self-loop~$(a,a)$
 	by an arc from~$r$ to~$a$ (again of weight~$0$).
 	For every pair of vertices~$a$ and $a'$ where $a$~disqualifies~$a'$, we create
 	an additional arc
 	from~$a$ to~$a'$ of weight~$\rho((a,a'))$.
 	For every agent~$a$ that does not initially qualify itself, we add an arc from~$r$ to~$a$
 	with weight~$\rho((a,a))$.
 	Finally, we mark all vertices in~$A^+$ as terminals.
 	We call the resulting graph~$G'$.

 	Now, it is easy to verify that there exists a subset of arcs~$E^*$ in $G'$
of total weight at most~$\ell$
 	ensuring that there is a directed path from the root~$r$ to every
terminal~$t \in T$ in $G'$ if and only if there is a successful bribery of
 	cost~$\ell$ for~$I$.

 	For the ``if direction'', consider a subset of arcs~$E^*$ of total weight
 	at most~$\ell$
 	ensuring that there is a directed path from the root~$r$ to every
terminal~$t \in T$ in $G'$.
 	Now, observe that bribing for every arc~$(a,a') \in E^*$ of nonzero weight
 	the agent~$a$ to qualify agent~$a'$ (and $a'$ to qualify itself if $a=r$)
is indeed of total cost~$\ell$ and
 	yields a
 	successful bribery, as a path from~$r$ to some terminal~$t$ implies that
the vertex~$t$ will be socially qualified.

 	For the ``only if direction'', consider a successful bribery for~$I$.
 	In particular, consider the graph~$G'$ constructed above for the WDST
 	instance and set the weight of
 	every arc corresponding to a bribed agent pair to zero.
 	Now, it is easy to verify that the modified WDST instance must have a
 	solution of total weight~$0$
 	(just consider an arbitrary spanning arborescence of minimum weight with
 	root~$r$).
 	Clearly, the set of arcs with modified weights that may be part of the
selected arborescence has weight at most~$\ell$
 	in~$G'$.

 	\medskip\noindent\textbf{CSR:} For \fcsr, we start by guessing an agent
$a^*$ that is qualified by
 	everyone (if there already exists such an agent, we pick it), bribe all
agents that currently disqualify $a^*$ to qualify it, and
 	substract the costs from the budget.
 	Then, we apply an algorithm very similar to the above case, where
 	the corresponding auxiliary graph is even a little bit simpler: In
particular, consider the qualification graph where we add one additional root
vertex~$r$.
 	Moreover, for each existing arc, set its weight to $0$ and remove every
self-loop.
 	For every pair of vertices~$a$ and $a'$ where $a$~disqualifies~$a'$, we create
 	an additional arc
 	from~$a$ to~$a'$ of weight~$\rho((a,a'))$.
 	For every agent~$a$ that is qualified by everyone (including at least~$a^*$), we add an arc from~$r$ to~$a$
 	with weight~$0$.
 	Finally, we mark all vertices from~$A^+$ as terminals and set~$p$ to~$\ell$.
 	We call the resulting graph $G'$.
 	The correctness proof works analogous to the \flsr case.
 \end{proof}

The hardness results for constructive bribery imply that
constructive+destructive bribery is also NP-hard and W[2]-hard with respect
to~$\ell$.  Utilizing a slightly more involved reduction from \textsc{Exact
Cover By 3 Sets}, it is even possible to show that the NP-hardness of
constructive+destructive bribery extends to the case where the briber is only
allowed to delete qualifications. This may be surprising, as destructive bribery
alone is polynomial-time solvable.

\begin{figure}[bt]
  \begin{center}
    \begin{tikzpicture}
      \node[vertex, label=180:$a^*$] (a*) at (0, 0) {};
      \node[vertex, label=0:$a'$] (a') at (5.5, 2.8) {};
      \node[vertex, label=0:$a''$] (a'') at (5.5, 3.6) {};

      \draw (2,2) rectangle (3.2,-2);
      \draw[fill=green!30] (2,2) rectangle (3.2,0.5);
      \draw[fill=green!30] (5,2) rectangle (6.2,-2);
      \draw (5,4) rectangle (6.2,2.4);
      \node (SS) at (1.5,1.3) {$\mathcal{S}'$};
      \node (S) at (2.5,-2.5) {$\mathcal{S}$};
      \node (U) at (5.5,-2.5) {$\mathcal{U}$};

      \draw[->,dashed,thick] (a*) -- node[above=3pt] {$2m$} (1.9,0);
      \draw[->,dashed,thick] (a*) -- (1.9,-1.5);

      \draw[->,dashed,thick] (3.3,1.3) --node[above left=3pt and 3pt] {$2m$} (a'');
      \draw[->,dashed,thick] (3.3,1.3) -- (a');

      \node (m) at (2.6,1.3) {$m$};
      \node (2m) at (2.6,-0.7) {$2m$};
      \node (A-) at (7.5,3.2) {$A^-=a'\cup a''$};
      \node (A+) at (7.5,0) {$A^+=\mathcal{U} \cup a^*$};

    \end{tikzpicture}
  \end{center}
  \caption{Visualization of the reduction from Proposition \ref{thm:fsr-cd-arc} for a successful bribery with budget $\ell=4m$. The $4m$ dashed arrows will be deleted in the bribery.
  The remaining arrows are omitted. The socially qualified agents after the bribery are agents in the green area together with $a^*$.
  }
  \label{fig:fsr-cd-arc}
\end{figure}
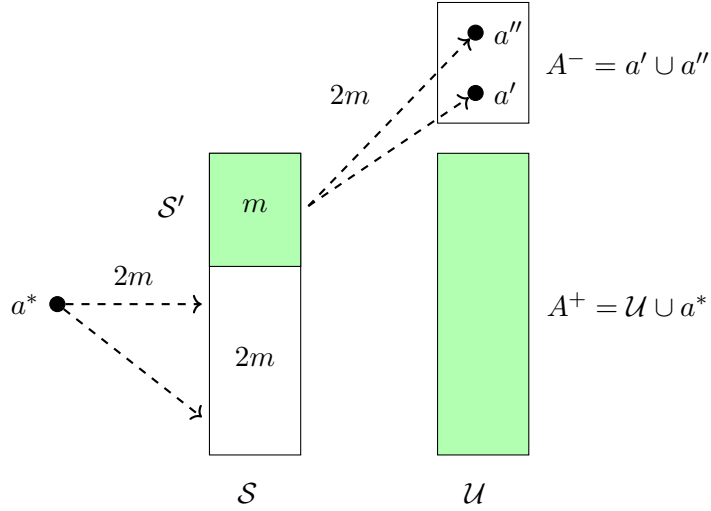

\begin{proposition}
	\label{thm:fsr-cd-arc}
	 \mixed \LinkBribery{\fcsr/\flsr} remains NP-complete even if the briber is only allowed to delete qualifications.
\end{proposition}
 \begin{proof}
 	We prove the proposition for both rules at the same time by a reduction from
 	\textsc{Exact Cover By 3 Sets}, a version of set cover where every set
contains exactly three elements and every element
 	appears in exactly three sets.  Given an instance
 	$(\mathcal{S},\mathcal{U})$ of \textsc{X3C}, where
 	$n=|\mathcal{U}|=|\mathcal{S}|=3m$, let $A=\mathcal{S}\cup \mathcal{U} \cup
 	a^{*} \cup a^{'} \cup a^{''}$, $A^+=\mathcal{U}\cup a^{*}$ and $A^-= a^{'}
 	\cup a^{''}$. We construct the qualification profile as follows. Everyone
 	qualifies $a^*$. Agent $a^*$ qualifies all agents from $\mathcal{S}$. All
 	agents
 	from $\mathcal{S}$ qualify $a^{'}$ and $a^{''}$ and the three agents
 	corresponding to their elements.
  We set the budget to $\ell = 4m$ and allow the briber only to delete
  qualifications. 	We now prove that there exists a successful bribery of cost
at most $\ell$
in the constructed instance if and only if there exists an exact cover in the
given \textsc{X3C} instance. See \autoref{fig:fsr-cd-arc} for a
 	visualization of the forward direction.

 	$(\Rightarrow)$ Assume that $\mathcal{S'}\subseteq \mathcal{S}$ is a cover
 	of $\mathcal{U}$ with $|\mathcal{S'}|=m$. Then, we bribe $a^*$ to
 	disqualify all agents from  $\mathcal{S}\setminus\mathcal{S'}$ (there are
 	$2m$ of them). Moreover, we make all agents from $\mathcal{S'}$ disqualify
 	$a'$ and $a''$ (there are $2m$ of them). In the resulting qualification
 	profile, all agents from $\mathcal{S'}$ are still socially qualified, while
 	all agents from $\mathcal{S}\setminus\mathcal{S'}$ are not socially
 	qualified. Thereby, all agents in $\mathcal{U}$ are socially qualified, as
 	they are qualified by one agent from $\mathcal{S'}$, while none of
 	$a'$ and $a''$ is socially qualified, as they are only qualified by members
 	of $\mathcal{S}\setminus \mathcal{S'}$.

 	$(\Leftarrow)$ Let us assume that we found a successful bribery of cost at
 	most $\ell$. Then, in the resulting qualification profile, the subset
 	$\mathcal{S'}\subseteq \mathcal{S}$ of agents which are still socially
 	qualified needs to cover all agents from $\mathcal{U}$. We claim that
 	$|\mathcal{S'}|\leq m$ and, in fact $|\mathcal{S'}|= m$, and thereby that
$\mathcal{S'}$ induces a solution of the
 	\textsc{X3C}-instance. For the sake of contradiction, assume that
 	$|\mathcal{S'}|>m$. In this case, to ensure that neither $a'$ nor $a''$ are
 	socially qualified, a budget of at least $|\mathcal{S}\setminus
 	\mathcal{S'}|$ was needed to prevent the existence of a path from $a^*$ to
 	$a'$ or $a''$ visiting
 	an agent from $\mathcal{S}\setminus
 	\mathcal{S'}$. Moreover, a budget of $2|\mathcal{S'}|$ was needed to make
 	all
 	agents in $\mathcal{S'}$ disqualify $a'$ and $a''$. This results in an
 	overall cost of $2|\mathcal{S'}|+|\mathcal{S}\setminus \mathcal{S'}|$,
 	which is larger than $4m$ as $|\mathcal{S'}|>m$.
 \end{proof}

The remaining question is to pinpoint the parameterized complexity of
constructive+destructive bribery with respect to $|A^+|+|A^-|$. Despite the
fact that the FPT result for constructive bribery suggests that it may be
possible to prove fixed-parameter tractability for this case, we were not even
able to prove that this
problem lies in XP, leaving this as an intriguing open problem for future work.

\section{Consent Rule} \label{se:cons}

For the consent rule, all computational problems in the link bribery setting
are easy to solve.
The main reasoning here is that every possible link bribery action can
influence the social qualification
of only one agent, namely, the sink of the corresponding arc.

For the consent rule, we distinguish four types of outcomes for some agent~$a$:
(1a) agent~$a$ qualifies itself and at least~$s-1$ other agents also qualify~$a$,
(1b) agent~$a$ qualifies itself and less than~$s-1$ other agents qualify~$a$,
(2a) agent~$a$ disqualifies itself and at least~$t-1$ other agents disqualify~$a$, and
(2b) agent~$a$ disqualifies itself and less than~$t-1$ other agents disqualify~$a$.
Moreover, an agent that shall be socially qualified must end up in case (1a) or (2b)
and an agent that shall be socially disqualified must end up in case (1b) or (2a).

For each of these cases, computing the cheapest bribery such that the case
applies for an agent $a$ in the modified profile is easy:
If this is not already the case, first bribe $a$ to (dis)qualify itself. Then,
sort the missing
(dis)qualifications of all other agents by their costs and iteratively pick the
cheapest one and make it (dis)qualify $a$ until
the condition stated in the specific case is met and $a$ becomes socially
(dis)qualified.
This way, for constructive bribery, we can compute for each agent from~$A^+$
the costs for the two cases (1a) or (2b) and choose the cheaper one, while for
destructive bribery we can compute for each agent from~$A^-$
the costs for the two cases (1b) or (2a) and choose the cheaper one.
As mentioned above, each agent can be handled independently.
Thus, we need to sort at most~$2(|A^+|+|A^-|)$ times at most~$|A|$~arc weights,
ending up with the following.

\begin{observation}
\label{ob:link-mixd-fst}
 \mixed \pLinkBribery{\fst} can be solved in $O((|A^+|+|A^-|)\cdot |A| \log |A|)$ time.
\end{observation}

Now, turning to agent bribery, we start by focusing on \const \pnAgentBribery 
\fst in Subsection \ref{sub:itclb} and explain how the results for constructive 
bribery translate to our 
other 
three bribery goals in Subsection \ref{sub:itdlb}. 

\subsection{Constructive Agent Bribery} \label{sub:itclb}
\citet[Theorem
2]{DBLP:journals/aamas/ErdelyiRY20} proved
that \const \AgentBribery \fst is NP-complete for all $s\geq 1$ and $t\geq 2$, 
which is why we focus on the parameterized complexity of this NP-complete 
problem.
Studying bribery problems for the consent rule, $s$ and $t$ are natural
parameters to consider, as at least one of these parameters may be small in
most
applications: In problems where socially qualified agents acquire a privilege,
$t$ should be small, while for
problems where social qualification implies some obligation or duty, $s$ should
be small. However, the hardness result of
\citet{DBLP:journals/aamas/ErdelyiRY20} directly
implies that  \const \AgentBribery \fst is para-NP-hard with respect to
$s+t$. However, the reduction
has no implications on the parameterized complexity of the problem with respect
to $\ell$ and $|A^+|$. In the following, we conduct a parameterized analysis of
\const \pnAgentBribery \fst with respect to $s$, $t$, $\ell$ and $|A^+|$, 
before, in Subsection \ref{sub:itdlb},
we explain how to adapt our results to the other three bribery goals 
considered. 

Interestingly, constructive agent bribery for consent rules can be seen as a
variant of a set cover problem: Every agent from $A^+$ that is not initially
socially qualified needs to gain a certain number of qualifications, that is,
the agent needs to be covered a certain number
of times, and bribing an agent $a\in A$ corresponds to covering all agents from
$A^+$ that agent~$a$ initially disqualified\footnote{The number of
qualifications an agent $a\in A^+$ that initially disqualifies itself needs to
gain obviously depends on whether $a$ is bribed (to qualify itself). However,
this can be
easily modeled by setting the initial demand of $a$ to $Q^-(a)-(t-1)$ and
insert $n-(t-1)-s$ copies of $a$ in the set that corresponds to bribing $a$. 
Formally, this is only possible if we have multisets that allow for multiple 
copies 
of an agent in one set. So we, in fact, need to consider a multiset cover 
variant.}.
In the following, we start by analyzing the parameterized complexity of \const
\pnAgentBribery{\fst} with respect to the parameters $s$, $t$,
and $\ell$. Subsequently, we analyze the influence of the number of agents that
we need to make
socially qualified, i.e., $|A^+|$, on the complexity of the problem.
\subsubsection{Parameterized complexity with respect to $s$, $t$, and $\ell$}
\citet{DBLP:journals/aamas/ErdelyiRY20} proved that \const
\AgentBribery{\fst} is in XP with respect to~$s$ if $t=1$. However, they left
open whether this problem is FPT or W[1]-hard. Moreover,
there exists a trivial brute force XP-algorithm for~$\ell$, while it is again
open whether the problem is FPT or W[1]-hard with respect to $\ell$. We answer
both questions negatively in the following theorem:

\begin{figure}[bt]
  \begin{center}
    \begin{tikzpicture}
      \node[vertex, label=180:$a_4$] (a4) at (0, 0) {};
      \node[vertex, label=180:$a_3$] (a3) at (0, 1) {};
      \node[vertex, label=180:$a_2$] (a2) at (0, 2) {};
      \node[vertex, label=180:$a_1$] (a1) at (0, 3) {};

    \node[vertex, label=270:$a_{1,2}$] (a12) at (2, 3) {};
     \node[vertex, label=270:$a_{1,3}$] (a13) at (2, 2) {};
      \node[vertex, label=270:$a_{1,4}$] (a14) at (2, 1) {};

      \node[vertex, label=270:$\tilde{a}_{1,2}$] (aa12) at (4, 3) {};
     \node[vertex, label=270:$\tilde{a}_{1,3}$] (aa13) at (4, 2) {};
      \node[vertex, label=270:$\tilde{a}_{1,4}$] (aa14) at (4, 1) {};

      \draw[->] (a1)--(a12);
      \draw[->] (a1)--(a13);
      \draw[->] (a1)--(a14);
      \draw[->] (a2)--(a12);
      \draw[->] (a3)--(a13);
      \draw[->] (a4)--(a14);
      \draw[->] (a12)--(aa12);
      \draw[->] (a13)--(aa13);
      \draw[->] (a14)--(aa14);
      \draw[->] (a12) to [loop above] (a12);
      \draw[->] (a13) to [loop above] (a13);
      \draw[->] (a14) to [loop above] (a14);
      \draw[->] (aa12) to [loop above] (aa12);
      \draw[->] (aa13) to [loop above] (aa13);
      \draw[->] (aa14) to [loop above] (aa14);

    \end{tikzpicture}
  \end{center}
  \caption{Visualization of the reduction from \autoref{thm:fst-c-agent-t=1-s}
for given graph $G=(\{1,2,3,4\},\{\{1,2\},\{1,3\},\{1,4\}\})$.
  }
  \label{fig:fst-c-agent-t=1-s}
\end{figure}
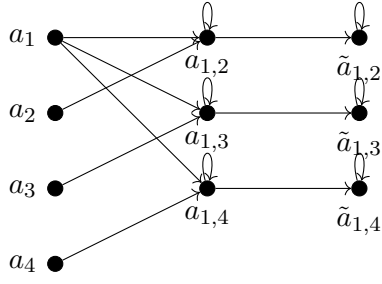
\begin{theorem}
    \label{thm:fst-c-agent-t=1-s}
    \const \AgentBribery \fst is W[1]-hard with respect to~$s+\ell$ even if~$t=1$.
\end{theorem}
\begin{proof}
    We reduce from \textsc{Independent Set} $(G=(V,E),
    k)$, where the question is to decide whether there exists a set of
    pairwise non-adjacent vertices of size $k$ in the given graph $G$.
    \textsc{Independent Set} is W[1]-hard with respect to $k$.
	Given an instance of the problem, in the group identification
	problem, we insert for each vertex $v\in V$, one \emph{vertex-agent} $a_v$
	and, for each edge $\{u,v\}\in E$, one \emph{edge-agent} $a_{u,v}$ and one
	designated \emph{dummy agent} $\tilde{a}_{u,v}$.
	We set $A^+=\{a_{u,v},\tilde{a}_{u,v}\mid \{u,v\}\in E\}$.
	All dummy agents qualify only themselves.
	For each $\{u,v\}\in E$, $a_{u,v}$ qualifies $\tilde{a}_{u,v}$ and itself, while it is only qualified by the two agents $a_u$ and $a_v$.
  See \autoref{fig:fst-c-agent-t=1-s} for a visualization.
	We set $s=k+2$, $t=1$ and $\ell=k$.
	Note that all dummy agents need at least $k$ additional qualifications to be
socially qualified, while all edge-agents need at least $k-1$ additional
qualifications. We now prove that there exists a successful bribery of cost
at most $\ell$
in the constructed instance if and only if there exists an independent set of
size $k$ in $G$.

	$(\Rightarrow)$ Assume that $V'\subseteq V$ is an independent set of size
	$k$ in $G$. Then, we bribe all vertex-agents that correspond to the agents
	in $V'$ and let them qualify everyone. Thereby, every edge-agent gains at
	least $k-1$ additional qualifications, as every edge is only touched by at most one
	vertex in $V'$ and thus each edge-agent was qualified by at most one agent
	$a_v$
	with $v\in V'$ before the bribery. Moreover, as no vertex-agent qualifies a
	dummy agent, all dummy agents gain $k$ qualifications.

	$(\Leftarrow)$ Assume we are given a successful bribery $A'\subseteq A$
consisting of $k$ agents.
	We claim that only vertex-agents can be part of $A'$.
	Assuming that for some $\{u,v\}\in E$ either $a_{u,v}$ or $\tilde{a}_{u,v}$
are part of $A'$, then $\tilde{a}_{u,v}$ cannot be socially qualified after the
bribery, as it gained at most $k-1$ qualifications over the bribery.
	It additionally holds that, for each $\{u,v\}\in E$,  at most one of $a_v$ and $a_u$ can be part of $A'$, as $a_{u,v}$ gained at least $k-1$ additional qualifications by the bribery.
	Thereby, $A'$ consists of $k$ vertex-agents which initially do not both qualify the same edge-agent.
	From this it follows that $V'=\{v\in V\mid a_v\in A'\}$ is an independent 
	set of size $k$ in $G$.
\end{proof}

\medskip 

\noindent{\textit{Remark 1}} \quad
	In the conference version of this paper 
	\citep{DBLP:conf/ijcai/BoehmerBKL20}, we observed that it is vital for the 
	correctness of the reduction that no non-vertex-agent can be bribed. This 
	is ensured by the existence of dummy agents who are all missing exactly 
	$\ell$ qualifications to become socially qualified. Motivated by this, we 
	wrongly claimed that 
	\const \AgentBribery \fst with $t=1$ can be solved in 
	$g(s)\mathcal{O}(n^2)$ for some function $g$ if the number of agents that 
	are missing exactly $\ell$ qualifications to become socially qualified can 
	be bounded in a function of $s$. In fact, this is not the case and the 
	reduction from \autoref{thm:fst-c-agent-t=1-s} can be easily adapted to 
	show W[1]-hardness even if only one agent is missing exactly $\ell$ 
	qualifications: We modify the construction by inserting an additional agent 
	$a^*$ who 
	qualifies everyone except itself and add $a^*$ to $A^+$. Further, we set 
	$s=k+3$, $t=1$, and $\ell=k+1$. Note that we need to bribe $a^*$ to qualify 
	itself and that afterwards all dummy agents are missing $k$ qualifications 
	to become socially qualified 
	and that the remaining budget is $k$.

\bigskip

From Theorem~\ref{thm:fst-c-agent-t=1-s}, it follows that combining $\ell$ with
parameters $s$ and $t$ is not enough to achieve fixed-parameter tractability.
Even fixing~$t$ as constant and considering $\ell$~and $s$~as parameter remains hard.
The only hope left when considering these three parameters is to treat $s$ as a constant.
Doing this, it turns that parameterizing the problem by $\ell$ alone is not enough to achieve
fixed-parameter tractability even in the restricted case where $s=1$.
This can be shown by a reduction from \textsc{Dominating Set} where given a graph
$G=(V,E)$ and an integer $k$, the task is to decide whether there exists a
subset of vertices $V'\subseteq V$ of size at most $k$ such that each vertex
$v\in V$ is either part of $V'$ or adjacent to at least one vertex from $V'$.
\textsc{Dominating Set} parameterized by $k$ is W[2]-hard.
\begin{theorem}
	\label{thm:fst-c-agent-l}
	\const \AgentBribery \fst is W[2]-hard with respect to~$\ell$ even if~$s=1$.
\end{theorem}
\begin{proof}
	We prove this theorem by a reduction from \textsc{Dominating Set}.  Given a
	undirected graph
	$G=(V,E)$, for $v\in V$, let $d(v)$ denote the degree of $v$
	in $G$ and let $d^*=\max_{v\in V} d(v)$ denote the maximum degree of a
	vertex in $G$. We construct a corresponding group identification problem as
	follows.

	For each $v\in V$, we introduce a so-called \emph{vertex-agent} $a_{v}$ who
	qualifies everyone except itself and the agents corresponding to its
	neighbors in $G$. Moreover, for each $v\in V$, we introduce $d^{*}-d(v)$
	dummy agents qualifying everyone except $v$. We set $A^+$ to be the set of
	all
	vertex-agents, $s=1$, $t=d^{*}+1$ and $\ell=k$. Initially, no vertex-agent
	is socially qualified, as each agent $a_v$ disqualifies itself and is
	disqualified
	by its $d(v)$ neighbors and $d^*-d(v)$ designated dummy agents. During the
	bribing, each agent needs to gain one additional qualification either from
	itself or from another agent to be included in the set of socially
	qualified agents. We now prove that there exists a successful bribery of
cost
at most $\ell$
in the constructed instance if and only if there exists a dominating set of
size at most $k$ in $G$.

	$(\Rightarrow)$ Assume that $V'\subseteq V$ is a dominating set of $G$.
Then,
	we bribe all vertex-agents corresponding to vertices in $V'$ such that they
	qualify everyone. Then, all agents in $V'$ are socially qualified as they
	qualify themselves. Moreover, for each agent $a_v$ with $v\in V\setminus
V'$, at least
	one of the vertex-agents corresponding to one of $v$'s neighbors no
	longer disqualify $a_v$. Thus, $a_v$ is disqualified by itself and by at
	most
	$d^*-1$ other agents and is, thereby, socially qualified.

	$(\Leftarrow)$ Assume that we have found a successful bribery bribing
agents
	$A'\subseteq A$. Then, without loss of generality, we can assume that no
	dummy agent is part of
	$A'$, as instead of bribing a dummy agent it is always possible to bribe
	the corresponding vertex-agent. Let $V'=\{v\in V\mid a_v\in A'\}$. As every
	vertex-agent $a_v$ needs to either qualify itself or gain an additional
qualification
	by at least one vertex agent corresponding to one of $v$'s neighbors
to be socially qualified in the end, $V'$
	is a dominating set of $G$.
\end{proof}

Parameterized by $\ell+t$, however, while keeping $s$ as a
constant, the problem becomes, in fact, fixed-parameter tractable:
\begin{theorem}
    \label{th:fst-lt-fpt}
    \const \pAgentBribery{\fst} is FPT with respect to  $\ell+t$ when $s$ is a constant.
\end{theorem}
\begin{proof}
We first prove the theorem for \const \AgentBribery{\fst} and afterwards
explain how to extend the algorithm for the priced version of this problem.
    Let $A'\subseteq A$ denote the subset of agents that we want to bribe. Note
that it is always optimal to bribe all agents from $A'$ to qualify everyone
(including themselves). We distinguish the case where $\ell<s$  and $\ell\geq
s$. In the former case, we simply iterate over all $\ell$-subsets of $A$,
bribe them to qualify everyone and check whether the bribery was successful.
This leads to an overall running time of
$\mathcal{O}(n^{s+2})$. It remains to consider the latter case where $\ell
\geq s$. Here, every agent qualifying itself will be socially qualified in the
end because we bribe all agents from $A'$ to qualify everyone and, thereby,
each agent will be qualified by at least $s$ agents after the bribery. Thus,
this problem
reduces to the case $s=1$ for which we describe the algorithm in the following.

For each $a\in A$, let
$y_a=\max(0,|Q^-(a)|-(t-1))$ denote the number of additional qualifications $a$
needs to get to become socially qualified without qualifying itself. For every
$a\in A^+$ with $y_a>\ell$, it needs to hold that $a\in A'$. More generally, for
all $a\in A^+$, either $a\in A'$ or a $y_a$-subset of $Q^-(a)$ needs to be part
of $A'$. These ideas give rise to Algorithm \ref{alg:algorithm2}, which we call
as CalcB($A$, $\varphi$, $A^+$, $\emptyset$, $\ell$). This
algorithm identifies at each step all agents that need to qualify
themselves because the remaining budget is not sufficient otherwise and all
agents that are already socially qualified after bribing the agents from $A'$.
Subsequently, an arbitrary agent  from $A^+$ which is not already socially
qualified after
bribing the agents from $A'$ is selected and we branch over bribing it or one
of the agents disqualifying it.

If the algorithm returns a solution, this solution is clearly correct, as
agents from $A^+$ are only deleted if they become socially qualified after
bribing the agents from $A'$ and we correctly keep track of the remaining
budget. Thus, it remains to prove that if there exists at least one
solution, the
algorithm will always find one: Assume that $S\subseteq A$ is a subset of
agents such that after making all agents from $S$ qualify everyone, all agents
from $A^+$ are socially qualified. We
now argue that there exists at least one branch in the execution of the
algorithm where it always holds that $A'\subseteq S$. We prove this by
induction on the size of $A'$ during the run of the algorithm. Before inserting
any agents in $A'$, the condition is trivially fulfilled. Let us assume that
$A'$ contains $i$ agents and it holds that $A'\subseteq S$. There
exist two possibilities where the $i+1$th agent $\tilde{a}$ is added to $A'$ in
the algorithm. Assuming that $\tilde{a}$ is added to $A'$ in line 2, then it
needs to hold that $\tilde{a}$ is also part of $S$ because assuming that all
agents from $A'$ are bribed, the remaining budget is not high enough to bribe
enough agents such that $\tilde{a}$ gets socially qualified in the end without
qualifying itself. Otherwise, $\tilde{a}$ gets inserted to $A'$ in line 9.
In line 9, we branch over inserting a candidate $a^*$ who is still missing at
least one qualification to get socially qualified after bribing the agents from
$A'$ and bribing one of the agents disqualifying $a^*$. As we have assumed
that $A'\subseteq S$, it needs to hold that either $a^*$ or an agent
disqualifying $a^*$ is part of $S\setminus A'$, as otherwise $a^*$ cannot be
socially qualified after bribing the agents from $S$. Thus, for at least one
branch, it still holds that $A'\subseteq S$ after inserting the
$i+1$th agent in $A'$.

The depth of the recursion is bounded by $\ell$. Moreover, the branching factor
is bounded by $|Q^-(a^*)|+1$. As for every
$a^*$ it needs to hold that $y_a\leq
\ell$ (because otherwise $a^*$ would have been deleted from $A^+$ in line
2), it follows that $|Q^-(a^*)|\leq \ell +(t-1)$. Thereby, the overall
running time of the algorithm lies in $\mathcal{O}(n^2(\ell+t)^\ell)$.

     	\begin{algorithm}[t]
 		\caption{CalcB($A$, $\varphi$, $A^+$, $A'$, $p$, $y_{a_1}$, \dots,
$y_{a_n}$)}
 		\label{alg:algorithm2}
 		\begin{algorithmic}[1] 
 		\Require  Agents $A$, qualification profile $\varphi$, subset of agents
 		$A^+$ and $A'$, remaining budget $p$
 		\Ensure Set of agents $A'$ to bribe
 		\For{$a\in A^+$ with $y_a-|A'\cap Q^-_{\varphi}(a)|>p$}
             \State $A'=A'\cup \{a\}$; $A^+=A^+\setminus \{a\};$ $p=p-1$;
        \EndFor
         \For{$a\in A^+$ with $a\in A'$ or $|A'\cap Q^-_{\varphi}(a)|\geq y_a$}
             \State $A^+=A^+\setminus \{ a\}$;
        \EndFor
 			\IIf {$p<0$} \textbf{return} Reject
 			\EndIIf
 			\IIf {$A^+=\emptyset$}
 			\textbf{return} $A'$
 			\EndIIf

 			\State Pick an arbitrary $a^*\in A^+$
 			\For{$a\in \{a^*\}\cup (Q^-_{\varphi}(a^*)\setminus A')$}
             \State \textbf{return} CalcB($A$,$\varphi$, $A^+$, $A'\cup \{a\}$,
$p-1$, $y_{a_1}$, \dots,
$y_{a_n}$)
             \EndFor
 	 		\end{algorithmic} \label{alg}
     	\end{algorithm}

We now show the theorem for \const \pAgentBribery{\fst}.
We still distinguish the case where $\ell < s$ and $\ell \ge s$.
In the former case, we simply iterate over all subsets of $A$ that cost at most $\ell$,  bribe them to qualify everyone and check whether the bribery was successful.
Since each agent has a positive integer price, the size of subsets of $A$ that
cost at most $\ell$ is upper bounded by $\ell<s$, and hence, the running time
is $\mathcal{O}(n^{s+2})$.

In the latter case where $\ell \ge s$, compared to the unpriced version, it
does not hold anymore that every agent qualifying itself will be socially
qualified in the end because with budget $\ell$ it is not guaranteed that we
can bribe at least $s$ agents.
So the main problem for \const \pAgentBribery{\fst} is how to guarantee that if an agent $a \in A^+$ is bribed, then $a$ will be socially qualified at the end.
Let $A^+_{\le \ell}=\{a \in A^+ \mid y_a \le \ell\}$ be the set of agents from
$A^+$ that need at most $\ell$ additional qualifications to become socially
qualified and $A^+_{> \ell}=\{a \in
A^+ \mid y_a > \ell\}$ the set of agents from $A^+$ that need more than $\ell$
additional qualifications to become socially qualified.

We distinguish the case where $n\leq s+\ell + t$ and where $n> s+ \ell
+ t$. In the former cases,  since $s$ is a constant, we have
$n=\mathcal{O}(\ell+t)$ and we can brute force through all subsets of $A$ that
cost at most $\ell$, bribe them to qualify everyone and check whether the
bribery was successful.
This leads to an overall running time of $\mathcal{O}((\ell+t)^22^{\ell+t})$.

In the latter case, a more involved approach is needed. We claim that for every
$a \in A^+_{\le \ell}$, if $a$ qualifies itself, then $a$ will be socially
qualified.
Notice that for $a \in A^+_{\le \ell}$, $n=|Q^+(a)|+|Q^-(a)|\le
|Q^+(a)|+y_a+(t-1) \le |Q^+(a)|+\ell+t$. Thus, as we are in the case with $n>
s+\ell  + t$, we can assume that for every $a \in A^+_{\le \ell}$, $|Q^+(a)|
\ge s$, which implies that $a$ will be socially qualified if $a$ qualifies
itself.
Next, for agents in $A^+_{> \ell}$, we have to bribe all of them because the budget is not enough to make them socially qualified if any one of them is not bribed.
However, after doing so, it is not necessarily the case that all agents from
$A^+_{>
\ell}$ are already socially qualified.
To ensure that all agents from $A^+_{>
	\ell}$ are socially qualified, we can iterate over all subsets of $A$ of
	size at most $s$, bribe them and then
apply Algorithm \ref{alg:algorithm2} with appropriate adjustments. To
summarize, this leads to the following algorithm:
First, we bribe all agents in $A^+_{> \ell}$ that do not qualify themselves,
set $A'$ to be the set of these
agents and modify the budget $\ell=\ell-\rho(A')$.
If all agents from $A^+_{> \ell}$ are already socially qualified after this
bribery, we
call
Algorithm \ref{alg:algorithm2} as CalcB($A$, $\varphi$, $A^+_{\le \ell}$, $A'$,
$\ell$)
with the following minor modifications: Delete
lines 1 and 2 (as we already know in advance that the branching factor in line
8 is bounded), and in line 9 return CalcB($A$, $\varphi$, $A^+$, $A'\cup
\{a\}$,
$p-\rho(a)$).
If not, we iterate over all subsets of $A$ with size at most $s$ and with cost
at most $\ell$, bribe them and check whether all agents in $A^+_{> \ell}$ are
socially qualified afterwards.
For each subset such that after bribing agents from this subset all agents from
$A^+_{> \ell}$ are socially qualified, we do the following: We add the agents 
from this subset to $A'$ and modify the
budget $\ell$
accordingly, and apply Algorithm \ref{alg:algorithm2} as CalcB($A$,
$\varphi$, $A^+_{\le \ell}$, $A'$, $\ell$) with the same
modifications described above.

If the above algorithm returns a solution, this solution is clearly correct.
Conversely, assume that $S\subseteq A$ is a subset of agents such that after making all agents from $S$ qualify everyone, all agents from $A^+$ are socially qualified.
First of all, it obviously needs to hold that $A^+_{> \ell} \subseteq S$.
Next, let $s'=\min\{|S \setminus A^+_{> \ell}|, s\}$ and let $S'$ be a set of agents such that $A^+_{> \ell} \subseteq S' \subseteq S$ and $|S' \setminus A^+_{> \ell}|=s'$.
It is easy to verify that $S'$ is a subset of $A$ with size at most $s$ and with cost at most $\ell$, and all agents in $A^+_{> \ell}$ will be socially qualified if we bribe all agents in $S'$.
Therefore, the above algorithm will apply Algorithm \ref{alg:algorithm2} after bribing
all agents in $S'$ and Algorithm \ref{alg:algorithm2} will find $S$ (or another
successful bribery containing $S'$).

Iterating over all subsets of $A$ with size at most $s$ and checking whether
all agents from $A^+_{> \ell}$ are socially qualified costs
$\mathcal{O}(n^{s+2})$.
Therefore, the overall running time lies in $\mathcal{O}(n^{s+4}(\ell+t)^\ell)$.
\end{proof}

\subsubsection{Parameterized complexity with respect to $|A^+|$}
Finally, analyzing the influence of the number of agents that should be made
socially qualified on the complexity of the problem, it turns out that
restricting this parameter is more powerful than restricting $\ell$, as the
problem is FPT with respect to~$\left| A^+ \right|$ for arbitrary $s$ and $t$.
We solve the problem by constructing an Integer Linear Program (ILP).
\begin{theorem}
	\label{thm:fst-c-agent-sizeAplus}
	\const \pAgentBribery{\fst} is FPT with respect to~$\left| A^+ \right|$.
\end{theorem}
\begin{proof}
	We first guess the subset $\widetilde{A}\subseteq A^+$ of agents from
	$A^+$ which we want to bribe. We bribe all agents from $\widetilde{A}$
    to qualify everyone (including themselves) and adjust $\ell$ and $\varphi$,
    accordingly.
	Note that in agent bribery it is never rational to bribe an agent from $A^+$ to disqualify itself, as we assume that  $s+t \le n+2$, and hence an agent needs less or as many qualifications ($s-1 \le n-(t-1)$) from other agents to be socially qualified if it qualifies itself compared to if it disqualifies itself.
	Now the problem is how to bribe the agents in $A \setminus A^{+}$ to make
	all agents in $A^{+}$ socially qualified.
	We reduce this to an ILP with $\left| A^+\right| +1 $ constraints.
	For every agent~$a \in A \setminus A^{+}$, we introduce a variable $x_a \in
\{0,1\}$ (with $x_a = 1$ if and only if we are going to bribe~$a$).
	We have two types of conditions for agents in $A^+$ based on whether they qualify themselves:
	\begin{align*}
	\sum_{a' \in A\setminus A^{+}: \varphi(a',a) = -1} x_{a'} &\geq s -
|Q^+_{\varphi}(a)|
&\forall a \in A^{+}: \varphi(a,a)=1  \\
	\sum_{a' \in A\setminus A^{+}: \varphi(a',a) = -1} x_{a'} &\geq
|Q^-_{\varphi}(a)| - (t - 1)    & \forall a  \in A^{+}: \varphi(a,a)=-1
	\end{align*}
	Finally, we set the objective function
	\[
	  \min \sum_{a \in A\setminus A^{+}} \rho(a) \cdot x_a \,.
	\]
	If the above ILP is feasible and the value of the objective function is at most $\ell$, then there exists a successful bribery (given
by $x_a = 1$): We bribe all $a\in A\setminus A^{+}$ with $x_a=1$ to qualify
everyone.
	In the bribed instance, for every agent $a \in A^{+}$ with
$\varphi(a,a)=1$, we have that at least $s$ agents
qualify~$a$ making $a$ socially qualified, and for every agent $a \in A^{+}$
with
$\varphi(a,a)=-1$, at most $t - 1$ agents disqualify~$a$ making $a$ socially
qualified.
If the above ILP is not feasible, then there is clearly no successful bribery
which
bribes exactly the agents $\widetilde{A}$ from $A^+$.
We can guess all $2^{|A^{+}|}$ subsets of $A^{+}$, and for each guess use the
algorithm of \citet{DBLP:conf/soda/EisenbrandW18} to solve the corresponding
ILP in time $2^{\mathcal{O}(|A^+|^2)}$, since the constraint matrix is indeed 
binary (and contains both lower and upper-bounds on variables).
\end{proof}
\subsection{Agent Bribery for Other Bribery Goals}  \label{sub:itdlb}
To apply results obtained for \const \pnAgentBribery{\fst} to the other
considered bribery goals, we use \citet[Lemma~1]{DBLP:journals/aamas/ErdelyiRY20}:
\begin{lemma}[\cite{DBLP:journals/aamas/ErdelyiRY20}] \label{le:1}
	$f^{(s,t)}(\varphi,A)=A\setminus f^{(t,s)}(-\varphi, A)$, where $-\varphi$
	is obtained from $\varphi$ by flipping all values of $\varphi$.
\end{lemma}

By \autoref{le:1},  results from above
naturally extend to \dest \AgentBribery{\fst} where $s$ and $t$ switch
roles:
\begin{corollary}
	\dest \AgentBribery{\fst} is W[1]-hard with respect to $t+\ell$ even if
	$s=1$ and W[2]-hard with respect to $\ell$ even if $t=1$. \dest
	\pAgentBribery{\fst} is FPT with respect to $\ell+s$ (treating $t$ as a
	constant) and FPT with respect to $|A^-|$.
\end{corollary}

Turning to constructive+destructive and exact bribery, recall that 
\citet[Theorem
2]{DBLP:journals/aamas/ErdelyiRY20} proved
that \const \AgentBribery{\fst} is NP-complete for all $s\geq 1$ and $t\geq 2$
by a reduction from \textsc{Vertex Cover} in which they set $A^+=A$. Applying 
\autoref{le:1}, this 
implies that \exact/\mixed \AgentBribery{\fst} is NP-hard as soon as 
either $s>1$ or $t>1$. Observe that in the only remaining case, i.e., $s=t=1$, 
the 
social qualification of an agent only depends on its opinion about itself. 
Thus, both  \exact/\mixed \AgentBribery{\fst} are 
linear-time solvable in this case:
\begin{observation}
	\label{ob:agent-mixed-fst}
	\exact/\mixed \AgentBribery{\fst} is linear-time solvable for~$s=t=1$. For 
	all other values of~$s$ and $t$, this problem is NP-complete.
\end{observation}

Moreover, it is also possible to extend some of our results to
constructive+destructive bribery by a slight adaption of Theorem
\ref{thm:fst-c-agent-t=1-s} and Theorem \ref{thm:fst-c-agent-sizeAplus}:
\begin{corollary}
\label{cor:fst-c+d}
	\mixed \pAgentBribery \fst is FPT with respect to~$\left| A^+ \right|+\left| A^- \right|$. \mixed \AgentBribery \fst parameterized by~$\ell+t$ is W[1]-hard even if~$s=1$ and also W[1]-hard with respect to~$\ell+s$ even if $t=1$.
\end{corollary}
 \begin{proof}
 	To prove the first part, we can use a similar approach as in the proof of Theorem~\ref{thm:fst-c-agent-sizeAplus}.
 	We first guess the subset $\widetilde{A}^+\subseteq A^+$ of agents from $A^+$ and the subset $\widetilde{A}^-\subseteq A^-$ of agents from $A^-$ which we want to bribe.
 	We bribe all agents in $\widetilde{A}^+$ to qualify all agents in $A^+$ (including themselves) and disqualify all agents in $A^-$.
 	Similarly, we bribe all agents in $\widetilde{A}^-$ to qualify all agents in $A^+$ and disqualify all agents in $A^-$ (including themselves).
 	The remaining problem is how to bribe agents in $A \setminus (A^+ \cup A^-)$, which can be reduced to an ILP with $|A^+ \cup A^-|+1$ constraints analogously to the ILP in the proof of Theorem~\ref{thm:fst-c-agent-sizeAplus}.

 	The second part holds by using Theorem \ref{thm:fst-c-agent-t=1-s} as well as
 	by using Theorem \ref{thm:fst-c-agent-t=1-s} but with switched roles of~$s$ 
 	and~$t$.
 \end{proof}
Aiming for positive results, parameterizing constructive+destructive bribery by
just one of $|A^+|$ and $|A^-|$ is not enough, as the problem is even NP-hard
for $s=2$ and $t=1$ and $A^+=\emptyset$ (and thus also for $s=1$, $t=2$ and
$A^-=\emptyset$), which follows from the reduction from
\citet{DBLP:journals/aamas/ErdelyiRY20} mentioned at the beginning of this
section.

Finally, we consider the exact bribery. Here, the two W[1]-hardness results from
Corollary \ref{cor:fst-c+d}, which follow from Theorem
\ref{thm:fst-c-agent-t=1-s},
still hold, as it is possible to precisely specify the desired
outcome of the group identification problem in the reduction in Theorem
\ref{thm:fst-c-agent-t=1-s}. While parameterizing the problem by~$|A^+|+|A^-|$
is
not meaningful in this context, as discussed above, one of $|A^+|$ and $|A^-|$
combined with $s$ and $t$ is not enough to achieve any positive
results:
\begin{corollary}
	\exact \AgentBribery \fst parameterized by~$\ell+t$ is W[1]-hard even
	if~$s=1$ and also W[1]-hard with respect to~$\ell+s$ even if $t=1$.
\end{corollary}

\section{Conclusion}\label{se:conc} 
In this paper, we extended the research on bribery in group identification by
considering priced bribery, a new model for bribery cost (which was independently
introduced by \citet{DBLP:journals/aamas/ErdelyiRY20}), and two new bribery goals.
Moreover, we described how it is possible to use bribery as a method to
calculate the margin of victory or distance from winning of agents in a group
identification problem.
We showed that for all considered rules both quantities can be computed
efficiently if the number of agents whose social qualification we are
interested in is small.
Moreover, we identified some further cases where the general problem is
polynomial-time solvable. Namely, for both the liberal-start-respecting rule
and the consensus-start-respecting rule, we observed that every bribery
involving destructive constraints splits the qualification graph into
two parts.
As in agent bribery it is easy to make multiple agents socially qualified at
the same time, this observation gives rise to a polynomial-time algorithm for
the most general variant of priced constructive+destructive bribery for agent
bribery.
For link bribery, it is possible to use this observation to construct a
polynomial-time algorithm for priced destructive bribery.
However, the corresponding question for constructive bribery is NP-hard.

For the consent rule, link bribery turns out to be solvable in a straightforward way.
In contrast to this, for agent bribery, finding an optimal bribery corresponds to finding a set of agents fulfilling certain covering constraints. 
This task makes the problem para-NP-hard with respect to $s+t$ \citep{DBLP:journals/aamas/ErdelyiRY20} and W[1]-hard with respect to $\ell$ even for $t=1$ or $s=1$.
Recall that $\ell$ is probably the most canonical parameter---the budget of the 
bribery which is either the number or the total costs of the bribery actions.
We proved that in the constructive setting, the rule parameter $t$ is slightly more powerful than the rule parameter $s$ in the sense that the problem is still hard parameterized by $s+\ell$ even if $t=1$, while becoming fixed-parameter tractable parameterized by $t+\ell$ for constant $s$.
Thus, we can see that the complexity implications of the two rule-specific
parameters $s$ and
$t$ are asymmetric in this case.

Finally, we would like to stimulate further research in this area with some open problems arising from our work.
First, we focused on the design of new FPT-algorithms for which we used
techniques ranging from rather basic ones (such as branching tree analysis) to
the use of ILPs.
A natural open question striking from this is whether some combination of
parameters for which we construct an FPT-algorithm can be improved to yield a
polynomial-sized kernel.
Another important direction is to identify new parameters for group
identification (mainly for consent rules).
A possible example of which could be the (edge/vertex) distance  of a
symmetrisation of the qualification graph to a cluster graph.
Yet another natural example is the Kendall-tau distance to a master (or
central) profile, which may arise in the cases when it is possible to ``test''
someones qualification. Moreover, it is also natural to consider the maximum
number $\Delta$ of agents that some agent qualifies as a parameter. In fact,
recently \citet{BBKL21a} showed that \const \AgentBribery{\fst} is
fixed-parameter tractable with respect to $\Delta+s$ for $t=1$ and
fixed-parameter tractable with respect to the combined parameter
$s+t+\ell+\Delta$, which complements our W[1]-hardness result for
this problem for the combined parameter $s+t+\ell$ from
\autoref{thm:fst-c-agent-t=1-s}.

\section*{Acknowledgments}
NB was supported by the DFG project MaMu  (NI 369/19).
DK is partly supported by the OP VVV MEYS funded project
CZ.02.1.01/0.0/0.0/16\_019/0000765 ``Research Center for Informatics''; part of
the work was done while DK was affiliated with TU Berlin and supported by
project MaMu (NI 369/19).
Part of the work was done while JL was affiliated with TU Berlin and supported by the DFG project AFFA (BR 5207/1 and NI 369/15).
This work was started at the research retreat of the TU Berlin Algorithms and Computational Complexity group held in September 2019.

\end{document}